\documentclass[runningheads]{llncs}
\usepackage[T1]{fontenc}
\usepackage{graphicx}
\usepackage{datatool}
\usepackage{complexity}
\usepackage{newfile}
\usepackage{refcount}
\usepackage{cite}
\usepackage{complexity}
\usepackage{framed}
\usepackage[english]{babel}
\usepackage{amsmath,amsfonts,mathtools}
\usepackage{color}
\usepackage{graphicx}
\usepackage{subcaption}
\usepackage{url}
\usepackage{authblk}
\usepackage{relsize}
\usepackage{graphicx}
\usepackage{algorithm}
\usepackage[noend]{algpseudocode}
\usepackage{tikz}
\usetikzlibrary{automata, positioning, arrows.meta}
\usepackage{pgfplots}
\usepackage{multirow}
\usepackage{ntheorem}
\usepackage{stackengine,scalerel}
\usepackage[strings]{underscore}
\usepackage{paralist}
\usepackage{booktabs}
\usepackage{mathtools}
\usepackage{scalerel}
\usepackage{csquotes}
\usepackage[capitalise]{cleveref}
\crefname{equation}{Eq.}{Eq.}
\crefname{figure}{Fig.}{Fig.}
\crefname{lemma}{Lem.}{Lem.}
\crefname{theorem}{Thm.}{Thm.}
\crefname{corollary}{Cor.}{Cor.}
\crefname{definition}{Def.}{Def.}
\crefname{claim}{Claim}{Claim}
\crefname{section}{Sec.}{Sec.}
\crefname{figure}{Fig.}{Fig.}
\crefname{table}{Tab.}{Tab.}
\crefname{algorithm}{Alg.}{Alg.}
\crefname{observation}{Obs.}{Obs.}
\crefname{example}{Ex.}{Ex.}
\crefname{remark}{Rem.}{Rem.}
\crefname{appendix}{App.}{App.}
\crefname{subfigure}{Fig.}{Fig.}
\usepackage{enumitem}
\usepackage{tikz}  % Required for creating diagrams
\usetikzlibrary{arrows.meta, positioning,arrows, automata, shapes}  % Optional: For advanced arrow tips and positioning
\usepackage{xcolor}  % Required for colored edges or text (if using \color or \color{colored})
\tikzset{auto, >= stealth}
\tikzset{every edge/.append style={thick, shorten >= 1pt, ->}}
\tikzset{initial/.style={draw, thick, <-, shorten <=1pt}}
\tikzset{player0/.style = {draw, thick, shape=circle, minimum size=0.5cm},node distance=1.5cm}
\tikzset{player1/.style = {draw, thick, shape=rectangle, minimum size=0.7cm},node distance=1.5cm}

\definecolor{customdarkgreen}{rgb}{0.0, 0.5, 0.0} % A darker shade of green

\usepackage{wrapfig}

\newcommand{\abs}[1]{\left\lvert{#1}\right\rvert}
\newcommand{\N}{\mathbb{N}}

\newcommand{\Q}{\mathbb{Q}}

\newcommand{\tail}{\hat{\tau}}
\newcommand{\fair}[1]{\branch^{#1}_{\texttt{fair}}}
\newcommand{\escape}[1]{\branch^{#1}_{\texttt{esc}}}
\newcommand{\simul}[1]{\branch^{#1}_{\texttt{sim}}}
\newcommand{\valu}[1]{\branch^{#1}_{\texttt{val}}}
\newcommand{\pos}[1]{\branch^{#1}_{\texttt{pos}}}
\newcommand{\zero}[1]{\branch^{#1}_{\texttt{0}}}

\newtheorem*{observation}{Observation}
\renewtheorem{claim}{Claim}
\usepackage[]{todonotes}

\renewcommand{\inf}{\mathtt{Inf}}
\newcommand{\plays}{\mathtt{plays}}
\newcommand{\play}{\mathtt{play}}
\newcommand{\infix}{H}
\newcommand{\energy}{\mathtt{En}}

\newcommand{\Win}{\mathtt{Win}}
\newcommand{\ESC}{\,\mathtt{ESC}}

\newcommand{\branch}{\mathtt{br}}

\newcommand{\meanpayoff}{\mathtt{MP}}
\newcommand{\meanVal}{\mathtt{avg}}

\newcommand{\floor}[1]{\left \lfloor #1 \right \rfloor}
\newcommand{\tree}{\mathcal{T}}
\newcommand{\ext}{\mathtt{ext}}
\newcommand{\pr}{\mathbf{p}}
\newcommand{\w}{\mathbf{w}}
\newcommand{\uu}{\mathbf{u}}

\begin{document}

% Set document status-----------
\label{beginningofdocument}
\newoutputstream{docstatus}
\openoutputfile{main.ds}{docstatus}
\addtostream{docstatus}{Work_in_progress}
\closeoutputstream{docstatus}
%-------------------------------

\title{Fair Quantitative Games}
\author{Ashwani Anand\orcidID{0000-0002-9462-8514} \and
Satya Prakash Nayak\orcidID{0000-0002-4407-8681} \and
Ritam Raha\orcidID{0000-0003-1467-1182} \and
Irmak Sa\u{g}lam\orcidID{0000-0002-4757-1631} \and
Anne-Kathrin Schmuck\orcidID{0000-0003-2801-639X}
}
\authorrunning{Anand et al.}
\institute{Max Planck Institute for Software Systems, Kaiserslautern, Germany}
\maketitle              % typeset the header of the contribution
\begin{abstract}
    We examine two-player games over finite weighted graphs with quantitative (mean-payoff or energy) objective, where one of the players additionally needs to satisfy a fairness objective. The specific fairness we consider is called \emph{strong transition fairness}, given by a subset of edges of one of the players, which asks the player to take fair edges infinitely often if their source nodes are visited infinitely often. 
    We show that when fairness is imposed on player~1, these games fall within the class of previously studied $\omega$-regular mean-payoff and energy games. On the other hand, when the fairness is on player~2, to the best of our knowledge, these games have not been previously studied.  
    We provide gadget-based algorithms for fair mean-payoff games where fairness is imposed on either player, and for fair energy games where the fairness is imposed on player~1. For all variants of fair mean-payoff and fair energy (under unknown initial credit) games, we give pseudo-polynomial algorithms to compute the winning regions of both players. Additionally, we analyze the strategy complexities required for these games. Our work is the first to extend the study of strong transition fairness, as well as gadget-based approaches, to the quantitative setting. We thereby demonstrate that the simplicity of strong transition fairness, as well as the applicability of gadget-based techniques, can be leveraged beyond the $\omega$-regular domain. 

% \keywords{Fairness \and Quantitative Games \and Mean-Payoff Games \and Energy Games.}

%\keywords{First keyword  \and Second keyword \and Another keyword.}
\end{abstract}
\section{Introduction}
\raggedbottom
Games on graphs serve as a formal and effective framework for automatically synthesizing correct-by-design software in cyber-physical systems (CPS). In this setting, one player represents a controller aiming to ensure a high-level logical specification in response to the actions of an adversarial player representing the external environment. Two-player graph games abstract the strategic reactive behavior of autonomous systems, such as robots \cite{robots,robot2,kress2018synthesisForRobots_Review} or cars \cite{cars,AlthoffBelta_Review_FMCEforAutonomousDriving}, and can be employed to synthesize \emph{correct-by-design} strategies through the algorithmic process of reactive synthesis.
Such strategies function as logical controllers to ensure the required behavioral guarantees, ensuring the system meets its safety and performance goals over its strategic interactions with its environment. 

In the CPS context, these game-based abstractions have been augmented with additional information, such as assumptions about the players' strategic behavior induced by the underlying physical processes~\cite{alurbook,fairness}, and quantitative objectives that allow to optimize strategies w.r.t.\ given performance metrics~\cite{costcontrol}. These enhancements enable reactive synthesis algorithms to compute control strategies that are more aligned with the specific requirements of CPS applications. However, these extensions typically come with increased computational costs, hindering their practical applicability.
In this paper, we show that a particular subclass of such games
-- namely, \emph{energy} and \emph{mean-payoff games} under \emph{strong-transition-fairness constraints} -- distinctly posses favorable computational properties.
% compared to general synthesis methods that combine qualitative and quantitative goals. 
We refer to these games as \emph{fair energy} and \emph{fair mean-payoff games}.

\subsection{Background}

\noindent\textbf{Strong Transition Fairness}. Strong transition fairness~\cite{sifakis,fairness,baier} is defined over specific edges, referred to as \emph{fair edges}, and restricts strategies to use a fair edge infinitely often if its source vertex is visited infinitely often. Strong transition fairness constraints are less expressive than general fairness, typically expressed via a qualitative objective called the \emph{Streett objective} \cite{streett}. Yet, they naturally arise in areas such as resource management~\cite{CAFMR13}, abstractions of continuous-time physical processes for planning~\cite{CPRT03, DIRS18, AGR20} and controller synthesis~\cite{thistle1998control, NOL17, MMSS2021}. At the same time, games with strong transition fairness conditions offer more favorable computational properties compared to those with Streett fairness, making them a compelling subject of study.

Despite their strong motivation from CPS applications and their favorable computational properties, strong transition fairness has so far been considered only in the context of qualitative objectives, such as parity~\cite{irmak1,irmak2} or Rabin~\cite{banerjee23, ThejaswiniThesis23}, where it is shown to be computationally inexpensive, preserving the overall complexity of the game and inducing negligible computational cost. This paper shows that these property carry over to quantitative games.

\noindent\textbf{Energy \& Mean-Payoff Games.} Quantitative objectives, such as energy or mean-payoff, allow capturing strategic limitations induced by constrained resources, such as the power usage of embedded components, the buffer size of a networking element~\cite{stringmatching}, or string selections with limited storage~\cite{ZwickPaterson96}. % common in CPS applications.

\emph{Mean-payoff games}, introduced in~\cite{meanpayoff}, assign integer weights to edges, representing the payoffs received by player~1 (the controller) and paid by player~2 (the environment) when an edge is taken. These games are played over infinitely many rounds, with player~1 aiming to maximize the long run average value of the edges traversed, while player~2 aims to minimize this value.

\emph{Energy games}, introduced more recently in~\cite{CAHS03}, also assign integer weights to edges, this time representing energy gains or losses. In this setting, the controller's goal is to construct an infinite path where the total energy in every prefix remains non-negative, while the environment aims to prevent this. 
% These games offer a natural structure for modelling various real-world phenomena 

It is known that determining the winner in both energy and mean-payoff games is in $\NP \,\cap \,\coNP$ \cite{ZWICK1996343}. The state-of-the-art algorithms 
 for mean-payoff games have runtime $\mathcal{O}(n^2mW)$ (and a runtime of $\mathcal{O}(nmW)$ for the threshold problem), where $n$ is the number of nodes, $m$ is the number of edges and $W$ is the maximum absolute weight in the game arena~\cite{cominrizzi}. For energy games, the best algorithms are either deterministic with the run time $\mathcal{O}(\min(mnW, mn2^{n/2}\log W))$~\cite{FasterAlgorithmsforMeanPayoffBCDGR11,dorfman} or randomized with the run time $2^{\mathcal{O}(\sqrt{n \log n})}$~\cite{BSV04}.
%  \AKS{I've cut the details here quite a bit, feel free to bring them int the related work section.}

%The combination of quantitative and qualitative objectives has been explored since the introduction of energy games. 
\smallskip\noindent\textbf{Combining Quantitative \& Qualitative Objectives.}
Conceptually, \emph{fair energy} and \emph{fair mean-payoff games} considered in this paper combine quantitative (i.e., energy or mean-payoff) with qualitative (i.e.\ strong transition fairness) obligations to constrain the moves of the players in the resulting game. 

Combinations of quantitative and qualitative objectives have been studied for several variants of mean-payoff and energy objectives combined with general $\omega$-regular goals~\cite{CAHS03, BCHJ09, MeanPayoffParityCHJ05,Helouet2019,EnergyParityChatterjeeDoyen2010, PseudopolynomialMeanPayoffParityDJL18}. %Here, the quantitative objective models resource constraints (e.g., power consumption), while the qualitative one defines functional requirements (e.g., liveness or safety), and the system must satisfy both simultaneously.
% 
% Mean-payoff parity games were introduced in~\cite{MeanPayoffParityCHJ05} as a proof of concept for solving the combination of quantitative and qualitative objectives. They generalize mean-payoff games with combination of any $\omega$-regular objective, as all the $\omega$-regular objectives can be turned into parity objectives.
%Mean-payoff parity games~\cite{MeanPayoffParityCHJ05} demonstrate the combination of these objectives, generalizing games that combine mean-payoff with $\omega$-regular objectives. 
In particular, energy parity games were introduced in~\cite{EnergyParityChatterjeeDoyen2010}, demonstrating these games' inclusion in $\NP \,\cap\,\coNP$ and polynomial equivalence to mean-payoff parity games, previously studied in \cite{MeanPayoffParityCHJ05}. Further, people have studied the combination of quantitative objectives (optimization or threshold variant) with (not necessarily $\omega$-regular) Boolean constraints~\cite{BruyereFRR17, AlmagorKV16, ClementeR15, VelnerC0HRR15, BouyerMM14, AlmagorKK15, AlfaroFHMS05, AlmagorBK13, ChatterjeeD11, BohyBR14a, ChatterjeeD11a}. One example is the Beyond Worst-Case synthesis framework, introduced by Bruyère et al.~\cite{BruyereFRR17}, which combines Boolean guarantees (mean-payoff thresholds) with quantitative performance optimization (expected payoff).

It is known that strong transition fairness conditions can be translated into a classical Streett condition, i.e.,\ a subclass of (qualitative) $\omega$-regular objectives, by transforming each fair edge $e=(u,v)$ into a Streett pair $(\{u\},\{e\})$ in the original game arena. If only player~1 (the controller) is constrained by the fairness, this transformation results in an energy (resp. mean-payoff) Streett game. Following~\cite{EnergyMuCalculus2020}, Energy Streett games\footnote{This also gives a solution for mean-payoff Streett games via their polynomial reduction to energy games~\cite{EnergyParityChatterjeeDoyen2010}.} can be solved in one of the following ways:
\begin{compactitem}[$\triangleright$]
    \item Encode the Streett energy objective as a $\mu$-calculus formula and solve the formula using symbolic fixed-point computations: The computation takes $\mathcal{O}(n\cdot [b+1]^{\floor{\frac{d}{2}+1}})$ symbolic steps, where $b= (d+1) \cdot \left( ((n^2+n)\frac{n}{2}! -1)W\right)$ is an upper bound on the credit a Streett energy game requires, where $d$ is the number of Streett pairs, which in our case corresponds to the number of fair edges. 
    % In the worst case, 
    This complexity can be simplified to $\mathcal{O}\left((\frac{n}{2}! \cdot W)^{\frac{m}{2}}\right)$ by omitting the polynomial factors where $m$ is the number of edges; which is super-exponential.
    
\item Translate the energy Streett game into a Streett game by encoding the domain of energy levels in the state space. This explodes the state space by multiplying the number of states with the above-mentioned upper-bound $b$. Using state-of-the-art algorithms for Streett games~\cite{nirpnueli}, this gives an algorithm with worst-case super-exponential runtime of $\mathcal{O}((\frac{n}{2}!W)^{m+1}m!)$.
\end{compactitem}

If, however, player~2 (the environment) is constrained by strong-transition fairness, the above reduction to energy (resp. mean-payoff) Streett games does not apply as the energy or mean-payoff objectives are not symmetric for both players. In this case, the given energy (resp. mean payoff) game is enhanced with a \emph{fairness assumption} weakening the opponent. While such games have been investigated for qualitative objectives~\cite{irmak1}, this paper is the first to consider quantitative games under such fairness assumptions.

\subsection{Contributions}
Based on the above discussion, we distinguish between \emph{1-fair} and \emph{2-fair} energy (resp. mean-payoff) games, emphasizing which player~$i$s additionally constrained by the strong transition fairness condition. 
%Since mean-payoff and energy objectives are not symmetric for both players, this results in distinct challenges. % for solving the resulting games. 
% \emph{2-fair} variants of these games, i.e., with fairness conditions imposed on player~2, do not fall into previously studied categories.
We provide a unified framework to solve both classes of games via a novel, gadget-based approach. Our gadgets enable the translation of fair quantitative games into regular quantitative games on a linearly larger state space. This allows solving these games in pseudo-polynomial time, significantly improving the computational complexity for \emph{1-fair} quantitative games, compared to the naive approaches outlined above. The origins of these gadget-based techniques can be traced back to the works in ~\cite{stochparity} and~\cite{stochrabin}, where similar gadgets were employed to convert stochastic parity and Rabin games with almost-sure winning conditions into regular parity and Rabin games. In~\cite{irmak1}, a gadget-based technique was developed that reduces fair parity games to (linearly larger) regular parity games, demonstrating a quasi-polynomial run time for these games. Our work extends this approach to fair games with quantitative objectives. Concretely, our contributions are as follows.\\

\noindent\textbf{Fair Mean-Payoff Games:}
\begin{compactitem}[$\triangleright$]
    \item \emph{Complexity.} Using gadgets, we reduce (1- and 2-) fair mean-payoff games to regular mean-payoff games on a game arena with linearly larger number of nodes and edges, and a maximum absolute weight of $\mathcal{O}(n^2W)$, where $n$ is the number of nodes and $W$ is the maximum absolute weight in the fair game arena. Following the state-of-the-art algorithm of~\cite{cominrizzi}, we obtain a complexity of $\mathcal{O}(n^4mW)$ for solving the fair mean-payoff games for optimal value problem, and a complexity of $\mathcal{O}(n^3mW)$ for threshold problem. 
    \item \emph{Determinacy.} The above reduction establishes that fair mean-payoff games are determined irrespective of which player has fairness constraints.
    \item \emph{Strategies.} We show that, in 1-fair mean-payoff games, player~1 in general needs infinite memory strategies to win, although finite memory suffices for achieving suboptimal values. Memoryless strategies suffice for player~2. In 2-fair mean-payoff games, player~2 has finite memory winning strategies, whereas memoryless strategies suffice for player~1. 
\end{compactitem}

\noindent\textbf{Fair Energy Games:}
\begin{compactitem}[$\triangleright$]
    \item \emph{Determinacy.} In contrast to fair mean-payoff games, we show that fair energy games (with unknown initial credit) are not determined in general. In particular, 2-fair energy games are not determined, whereas 1-fair energy games are. We argue that the lack of determinacy prevents us from constructing similar gadgets to reduce 2-fair energy games to regular energy games.
    \item \emph{Complexity.} We introduce a gadget that reduces \emph{1-fair energy games} to regular energy games on an arena with linearly larger number of nodes and edges, and a maximum absolute weight of $\mathcal{O}(n^3W)$. Using the state-of-the-art algorithm of~\cite{dorfman}, we obtain a complexity of $\mathcal{O}(n^4mW)$ for solving these games. 
    For \emph{2-fair energy games}, we provide simple algorithms to compute the winning regions separately for both players. We show that in 2-fair energy games, player~1's winning region is the same as in the corresponding energy game where fair edges are treated as regular edges. Meanwhile, player~2's winning region is the same as in the corresponding 2-fair mean-payoff games with threshold $0$. These reductions yield an $\mathcal{O}(n^3mW)$ algorithm for solving 2-fair energy games.
    \item \emph{Strategies.} We show that in fair energy games, the player restricted by transition fairness has finite memory winning strategies whereas memoryless strategies suffice for the respective other player.
\end{compactitem}

\textbf{In summary}, our results show that strong transition fairness can be seamlessly incorporated into quantitative games without significant computational overhead. In fact, just as in the qualitative setting, this simple form of fairness comes \emph{virtually for free} in the quantitative setting. Our key conceptual insight is that both the simplicity of strong transition fairness, and the effectiveness of gadget-based approaches extend beyond the $\omega$-regular domain.

\section{Preliminaries}\label{sec:preliminaries}

% !TeX root = ../main.tex
In this section, we introduce the basic notations used throughout the paper. We write $\Q$ and $\N$ to denote the set of rational numbers and the set of natural numbers including $0$, respectively.

% $\Z$ 
% the set of integers; 
% $\Q$ the set of rational numbers and $\N$ the set of natural numbers 
% including 0.
% ; and $\N_{>0}$ the set of positive integers.

\subsection{Weighted Game Arena}	
A two-player weighted game arena is a tuple $G= (Q, E, w)$, where $Q= Q_1 \uplus Q_2$ is a 
finite set of nodes, $E \subseteq Q \times Q$ is a set of edges, and $w: E \to [-W, W]$ is a 
weight function that assigns an integer weight to each edge in $G$, where $W 
\in \N$ is the maximum absolute weight appearing in $G$. The nodes are partitioned into 
two sets, $Q_1$ and $Q_2$, where $Q_i$ is the set of nodes controlled by Player $i$ for $i \in 
\{1,2\}$ and $Q_1 \cap Q_2 = \emptyset$. 
For a node $q$, we write $qE$ to denote the set $\{e \in E \mid e = (q, q') \text{ 
for } q' \in Q\}$ of all outgoing edges from $q$ and $E(q)$ to denote all successors of $q$. 
W.l.o.g., we assume that all nodes 
have out-degree at least one, i.e., $E(q) \neq \emptyset$ for all $q \in Q$.
% $\forall q \in Q, \exists q' \in Q \text{ s.t. } (q,q') 
% \in E$. 
In rest of the paper, we use circles and squares in a figure to denote nodes controlled by player~1 and player~2, respectively.
	
\smallskip \noindent\textbf{Plays.} A \emph{play} $\tau = q_0 q_1 \ldots \in Q^\omega$ on $G$ is an infinite sequence of nodes starting from $q_0$ such that, $\forall i \geq 0, (q_i, q_{i+1}) \in E$. %We denote the $i^{th}$ node appearing in the play $\tau$ as $\tau[i]$. 
We use notations $\tau[0 ;i] = q_0 \ldots q_i$, and $\tau[i;j] = q_i \ldots q_j$ to denote the finite \emph{prefix} and \emph{infix} of the play $\tau$, respectively. A node $q$ is said to \emph{appear infinitely often} in $\tau$, i.e., $q \in \inf(\tau)$ if $\forall i, \exists j \geq i$, such that $q_j =q$. We naturally extend the notion of appearing infinitely often in a play for the edges. We let $\plays(G)$ denote the set of all plays on $G$, and let $\plays(G,q)$ denote the set of all plays starting from node $q$.

\smallskip \noindent\textbf{Strategies.} A \emph{strategy} $\sigma$ for player $i \in \{1, 2\}$ (or, a player~$i$-strategy) is a function  $\sigma: Q^* \cdot Q_i \mapsto E$ where for all $H \cdot q \in Q^* \cdot Q_i$, $\sigma(H \cdot q) \in qE$. Intuitively, from every node $q \in Q_i$, a strategy for player~$i$ assigns an outgoing edge of that node based on a history $H \in Q^*$. 

For a player~$i$ strategy $\sigma$, a play $\tau = q_0 q_1 \ldots$ is called a $\sigma$-play if it conforms with $\sigma$, i.e., for all $j \in \N$ with $q_j \in Q_i$, 
% if $q_j$ appears in $\tau$, then 
it holds that $\sigma(q_0 \ldots q_j) = (q_j,q_{j+1})$. 
Given a strategy $\sigma$ we denote the restriction of sets $\plays(G)$ and $\plays(G,q)$ to $\sigma$-plays with $\plays_{\sigma}(G)$ and $\plays_{\sigma}(G,q)$. Similarly, $\play_{\sigma, \pi}(G, q)$ denotes the unique play from $q$ conforming with player~1 and player~2 strategies $\sigma$ and~$\pi$.
%Given a strategy $\sigma$, we write $\plays_{\sigma}(G,q)$ to denote the set of all $\sigma$-plays starting from node $q$, and write $\plays_{\sigma}(G)$ to denote the set of all $\sigma$-plays in $G$. Similarly, we write $\play_{\sigma, \pi}(G, q)$ to denote the unique play that conform with a player~1 strategy $\sigma$ and player~2 strategy $\pi$ starting from $q$.

% Given any pair of strategies $(\sigma, \pi)$ of both the players, we naturally extend the notion of $\sigma$-plays to $\sigma, \pi$-plays and 
% use $\plays_{\sigma, \pi}(G,q)$ to denote the set of all $\sigma, \pi$-plays in $G$ starting 
% from node $q$. \ASin{$ \plays_{\sigma, \pi}(G,q) $ is always a singlton set. Once both 
% astrategies are fixed, there is only one play from $ q $.}\RIT{It is singleton only for memoryless strategies, right?}

Let $M$ be a set called \emph{memory}. A player~$i$ strategy $\sigma$ with memory $M$ can be represented as a tuple $(M, m_0, \alpha, \beta)$, where $m_0 \in M$ is the initial memory value, $\alpha : M \times Q \to M$ is the update function, and $\beta : M \times Q_i \to Q$ is the function prescribing next state. Intuitively, if the current node is a player~$i$ node $q$ and $m$ is the current memory value, the strategy $\sigma$ selects $q' = \beta(m, q)$ as the next node and updates the memory to $\alpha(m, q)$. If $M$ is finite, then we call $\sigma$ a \emph{finite memory strategy}; otherwise it is an \emph{infinite memory strategy}. 
Formally, given a history $H\cdot q\in Q^*\cdot Q_i$, $\sigma(H\cdot q) = 
\beta(\hat{\alpha}(m_0, H), q)$, where $\hat{\alpha}$ extends $\alpha$ to sequences of nodes 
canonically.
A strategy is called \emph{memoryless} or \emph{positional} if $|M| = 1$. For such memoryless 
strategy $\sigma$, it holds that $\sigma(H_1 \cdot q) = \sigma(H_2 \cdot q)$ for every history 
$H_1,H_2 \in Q^*$. We denote it as $\sigma(q)$ for convenience. For any positional strategy 
$\sigma$ of player~$i$, we define $G_\sigma$ to be the subgame arena $(Q,E',w)$ of $G$, where  
$E' = \{(q, \sigma(q)) \mid q \in Q_i\} \cup \{qE \mid q \in Q_{3-i}\}$. 
Note that in this subgame arena, each player~$i$ node has exactly one successor according to the strategy $\sigma$.

\smallskip \noindent\textbf{Weighted Games and Objectives.} A \emph{weighted game} is a tuple $(G, \varphi)$, where $G$ is a weighted game arena and $\varphi \subseteq Q^\omega$ is an \emph{objective} for player~$1$. 
A play $\tau$ is \emph{winning} for player~1 if $\tau \in \varphi$, else it is winning for player~2.
A player~$i$ strategy $\sigma$ is \emph{winning} from some node $q$, if all $\sigma$-plays starting from $q$ are winning for player~$i$.
We denote the winning regions of player~$i$ by $\Win_i(G, \varphi)$, or shortly $\Win_i$, when $(G, \varphi)$ is clear from the context. 
Formally, the winning regions are defined as follows where $\Sigma_i$ is the set of all player~$i$ strategies:
\begin{align}
   q \in \Win_1(G, \varphi) &\iff \exists \sigma\in\Sigma_1. \, \forall \pi \in \Sigma_2, \ \play_{\sigma, \pi}(G,q) \in \varphi \label{eq:determinacyLine1}\\
	q \in \Win_2(G, \varphi) &\iff \exists
    \pi\in\Sigma_2.\, \forall \sigma \in \Sigma_1,\ \play_{\sigma, \pi}(G,q) \not \in \varphi \label{eq:determinacyLine2}
\end{align}
A game $(G, \varphi)$ is called \emph{determined}, if $Q = \Win_1 \cup \Win_2$. 
% As the winning conditions of player~1 and 2 are the negation of one another, this definition makes sure that weighted games are determined, i.e. all nodes in $Q$ are either in $\Win_1$ or $\Win_2$. \ISin{An alternative definition would be... or should we give it below?}

Given a finite infix $\infix = q_0 \ldots q_k$ of a play in $G$, we denote the (total) weight and the average weight of $\infix$ by $w(\infix) = \sum_{i=0}^{k-1} w(q_i,q_{i+1})$ and $\meanVal(\infix) = \frac{w(\infix)}{k}$, respectively. 
Furthermore, we denote the (limit) average weight of a play $\tau$ by $\meanVal(\tau) = \liminf_{i \to \infty} \meanVal(\tau[0;i])$.

\smallskip \noindent\textbf{Mean-Payoff Games.}
A \emph{mean-payoff game} is a weighted game with a mean-payoff objective for a given threshold value $v \in \Q$, which is defined as $\meanpayoff_v = \{\tau \in Q^\omega \mid \meanVal(\tau) \geq v\}$. 
Intuitively, a play is winning for player~1 if the average weight of the play is above a certain threshold value $v$.
% the mean-payoff objective asks Player $1$ to keep the limit average weight of a play to be above a certain threshold.
% Player $1$ wins the mean-payoff game starting from $q$ if there exists a threshold value $v \in \Q$ and a strategy $\sigma$ such that $\plays_{\sigma}(G,q) \subseteq \meanpayoff_v$.
Note that, mean-payoff games are \emph{prefix independent} as any finite prefix of a play $\tau$ does not affect the limit average weight of $\tau$.

In this work, without loss of generality, we focus on mean-payoff objective $\meanpayoff_0$ with threshold value $v = 0$. It is easy to see that for any value $v$, the weighted mean-payoff game $(G, \meanpayoff_v)$ with $G = (Q,E,w)$ can be reduced to a mean-payoff game $(G', \meanpayoff_0)$ by subtracting $v$ from the weights of all edges in $G$, i.e., $G' = (Q,E,w')$ where $w'(q,q') = w(q,q') - v$ for all $(q,q') \in E$.

%In mean-payoff games, the winning region $\Win_2$ given by~\cref{eq:determinacyLine2} has the following equivalent formulation:
% \begin{align}
% q \in \Win_2(G, \varphi) &\iff \exists \pi\in\Pi.\,\forall
% \sigma\in\Sigma.\,\,\plays_{\sigma, \pi}(G,q) \not \subseteq \meanpayoff_v. \label{eq:determinacyLine2meanpayoff}
% \end{align}
% That is, we can swap the quantifiers $\forall$ and $\exists$, and obtain symmetric winning conditions for both players. \ISin{this swapping.. we call this.. it holds true for $\omega$-regular games but in weighted games not trivial...}

% \IS{Can we say: For convenience, with slight abuse of notation, we will define the mean-payoff of a play $\tau$ as $\MP(\tau)= \liminf_{i \to \infty} \frac{w(\tau,\cdot i)}{i}$, and the mean-payoff of a finite infix $\tau[i \cdot j]$ of a play as the average weight, $\MP(\tau[i \cdot j]) = \frac{w(i\cdot j)}{j-i+1}$. Or if this is too controversial of a notation, we can give this another name.}

The \emph{optimal value} of a mean-payoff game is the maximum value $v$ for which player~$1$ has a winning strategy with threshold $v$. A winning strategy of player~$1$ that achieves this optimal value is called an \emph{optimal value strategy}. The \emph{optimal value problem} is to compute the optimal values in a mean-payoff game.

\smallskip \noindent\textbf{Energy Games.} An \emph{energy objective} w.r.t.\ a given initial credit $c \in \N$ is defined as $\energy_c = \{\tau \in Q^\omega \mid c+ w(\tau[0;i]) \geq 0,\ \forall i \in \N\}$. Intuitively, a play belongs $\energy_c$ if the total weight (`energy level' starting from $c$) of the play remains non-negative along the play.
% Player $1$ wins the energy game starting from $q$ if there exists an initial credit $c \in \N$ and a strategy $\sigma$ such that $\plays_{\sigma}(G,q) \subseteq \energy_c$. 
An \emph{energy game} is a weighted game with an energy objective with unknown initial credit, denoted by $\energy$, where player~$1$ wins from some node $q$ if there exists an initial credit $c \in \N$ such that she can ensure the objective $\energy_c$ from $q$.
Formally, the winning regions are defined as follows, for $\varphi_c = \energy_c$:
\begin{align}
    q \in \Win_1 &\iff \exists c \in \mathbb{N}. \, \exists \sigma \in \Sigma_1. \, \forall \pi \in \Sigma_2,\ \play_{\sigma, \pi}(G, q)\in \varphi_c \label{eq:energyDeterminacyPlayer1}\\
	q \in \Win_2 &\iff \exists
    \pi\in\Sigma_2. \, \forall c \in \mathbb{N}. \,\forall \sigma \in \Sigma_1,\ \play_{\sigma, \pi}(G,q) \not \in \varphi_c \label{eq:energyDeterminacyPlayer2}
\end{align}

Note that, in contrast to mean-payoff games, the energy objective is not a \emph{prefix-independent} objective as the total weight of each prefix has to be non-negative along any play. 
However, it is known that such energy games are log-space equivalent to mean-payoff games~\cite{FasterAlgorithmsforMeanPayoffBCDGR11}.
Furthermore, it is known that both energy games and mean-payoff games are \emph{positionally determined}~\cite{meanpayoff,CAHS03} i.e., $Q = \Win_1 \cup \Win_2$ and both players have positional winning strategies.

\subsection{Fair Game Arena} A \emph{fair game arena} is a tuple $(G, E_f)$, where $G$ is a weighted game arena as defined above and $E_f \subseteq E$ is a given set of `fair edges'. We call the non-fair edges $E\setminus E_f$ `regular', and a node $q$ `fair' if it is the source node of a fair edge, i.e., $qE\cap E_f \neq \emptyset$. Let $Q_f$ be the set of all fair nodes in $G$. 
We investigate the scenario where all fair nodes are owned by the same player. If $Q_f \subseteq Q_1$, we call the game arena \emph{1-fair} and if $Q_f \subseteq Q_2$, we call it \emph{2-fair}.
For brevity, we sometimes denote the fair game arena $(G,E_f)$ by $G$ and keep the set of fair edges $E_f$ implicit. Given a fair node $q$, we write $qE_f$ to denote the set $\{e \in E_f \mid e = (q, q') \text{ 
for } q' \in Q\}$ of all fair outgoing edges from $q$ and $E_f(q) = \{q' \mid (q,q')\in E_f\}$ to denote the set of successors of $q$ via fair edges. 

A play $\tau$ is \emph{fair} if for all nodes $q \in \inf(\tau)$ and for all edges $e \in qE$, if $e \in E_f$, then $e \in \inf(\tau)$. We let $\plays^f(G)$ denote the set of all fair plays on $(G, E_f)$. 
We say a strategy $\sigma$ for player~$i$ is fair if $\plays_{\sigma}(G) \subseteq \plays^f(G)$. 

\smallskip \noindent\textbf{Fair Mean-Payoff/Energy Games.} An \emph{$i$-fair mean-payoff game} with objective $\varphi = \meanpayoff_v$ is a tuple $(G, E_f, \varphi^f)$, where $(G, E_f)$ is an $i$-fair game arena, and $\varphi^f$ is the objective with the respective fairness condition defined as follows:
\begin{equation} \varphi^f = \begin{cases}\varphi \cap \plays^f(G) &\text{ if } (G, E_f) \text{ is 1-fair} \\ 
  \varphi \cup \plays(G) \setminus \plays^f(G) &\text{ if } (G, E_f) \text{ is 2-fair}\end{cases}
\end{equation}
Intuitively, a play is winning for player~$i$ in $i$-fair game if it satisfies both the fairness condition and the player's corresponding objective, i.e., $\varphi$ for player~1 and $\neg \varphi$ for player~2.

Similarly, an \emph{$i$-fair energy game} with energy objective $\energy$ is a tuple $(G, E_f, \energy^f)$ where the winning regions are defined as in~\eqref{eq:energyDeterminacyPlayer1}-\eqref{eq:energyDeterminacyPlayer2} for $\varphi_c = \energy^f_c$.

\section{Example: Mean-Payoff vs Fair Mean-Payoff Games}

In this section, we depict the intricacy of fair mean-payoff games in comparison to regular mean-payoff games via the following example:

Consider the game arena depicted in \Cref{fig:ex-fair} with both the nodes belonging to player~1. With regular mean-payoff objective with threshold $0$, player~1 has a simple optimal positional strategy: at $q$, play the self-loop $(q,q)$ and from $p$ play the edge $(p,q)$. This strategy is winning for player~1 because the limit average weight of any play conforming to this strategy is $1$. Note that the strategy is also an optimal value strategy as the optimal value of the game is $1$.
\begin{wrapfigure}[6]
  {r}{0.25\textwidth}
  % \begin{figure}[h]
    \centering
    \vspace{-0.95cm}
    \begin{tikzpicture}
    
      \node[player0] (q) {$q$};
      \node[player0] (qprime) [right of=q] {$p$};
    
      \path[every node/.style={font=\sffamily\small}, ->, >=stealth]
        (q) edge [loop above] node {+1} (q)
            edge [bend left, dashed] node {-4} (qprime)
        (qprime) edge [bend left] node {0} (q);
    \end{tikzpicture}
    \vspace{-0.2cm}
    \caption{}\label{fig:ex-fair}
  % \end{figure}
  \end{wrapfigure}

  Now let us consider the game arena to be $1$-fair with a fair edge $(q,p)$. Now the fair mean-payoff objective $\meanpayoff_0^f$ for player~1 is to ensure that the limit average weight of the play is $\geq 0$ along with the fairness condition that says if $q$ appears infinitely often in a play, so does the edge $(q,p)$.

First note that there is no positional strategy that can ensure both fairness and mean-payoff objective and hence winning for player~1. Now, consider the following finite memory strategy for player~1: the strategy takes the self-loop $(q,q)$ for $k$ times before taking $(q,p)$ once, and this sequence of $k+1$ choices is repeated forever. The resulting play uses the fair edge $(q,p)$ infinitely often and hence it is fair. For $k \geq 4$, this strategy is winning as the limit average of the plays conforming to this strategy is at least $0$. In fact, we can show that for all $\epsilon > 0$, it is possible to choose a large enough $k$, such that player~1 ensures a limit average weight of at least $1-\epsilon$. 

Finally, consider the following infinite-memory strategy for player~1: the strategy is played in rounds; in round $i\geq 0$, the strategy plays $(q,q)$ for $i$ times, then plays $(q,p)$ once and progresses to round $i + 1$. This fair strategy ensures player~1 achieves a limit average weight of value $1$, and it is therefore an optimal value strategy.
However, it is not hard to see that there is no finite-memory optimal value strategy for player~1 in this game. In the following sections, we formally present how to solve fair mean-payoff games and discuss the memory requirements for each player to achieve these winning strategies.\label{sec:example}

%\section{Computational Complexity of Fair Energy Games}

%\input{sections/algo}

%\subsection{Fairness on Player 2}

%\input{sections/P2-fair-energy.tex}

% \section{Relationship with Fair Mean-payoff Games}

% \input{sections/mean-payoff}
\section{Solving Fair Mean-Payoff Games}\label{sec:solving-fair-mean-payoff}

% \ISin{Prove that
% \begin{itemize}
%   \item $\MP(\rho)  < 0 \Rightarrow \liminf_{i \to \infty} \energy(\rho_i) \to -\infty$,
%   \item  $\liminf_{i \to \infty} \energy(\rho_i) \not \to -\infty \Rightarrow \MP(\rho)  \geq 0$,
% \end{itemize}
% }
In this section, we present algorithms for solving fair mean-payoff games. The algorithms transform a fair mean-payoff game into a `regular' mean-payoff game using \emph{gadgets}. We introduce two gadgets depending on which player owns the fair edges as depicted in \cref{fig:gadgets}. 
% The circle nodes belong to player~1 and square nodes belong to player~2. %The gadget for 1-fair (2-fair) mean-payoff games is demonstrated in~\Cref{fig:Player1-gadget} (\Cref{fig:Player2-gadget}).
Given an $i$-fair mean-payoff game $G$, we replace all fair nodes with the corresponding $i$-fair gadgets and obtain an equivalent mean-payoff game $G'$ such that $(G, E_f, \meanpayoff^f_0)$ is winning for player~$i$ if the mean-payoff game $(G', \meanpayoff_0)$ is winning for player~$i$. In particular, we prove the following.

% We show that~\cref{thm:gadgetcorrectness} holds true regardless of which player owns the fair edges in the game.
% The solution for Player 1-fair games is realized via the gadget in \cref{fig:Player1-gadget}
% and the solution for Player 2-fair games is realized via the gadget in \cref{fig:gadget}.

\begin{theorem}\label{thm:mpgadgetcorrectness}
    Let $(G, E_f, \meanpayoff^f_0)$ be a fair mean-payoff game with threshold $0$, where $G = (Q, E, w)$, $w : E \to [-W, W]$, and $|Q|=n$. Then there exists a mean-payoff game $G'= (Q',E',w')$ where $Q' \supseteq Q$, $|Q'|\leq 6n$ and $w' : E' \to [-W', W']$ with $W' = n^2W + n$ such that
    $\Win_i(G) = \Win_i(G') \cap Q$ for $i \in \{1, 2\}$.
  \end{theorem}
 
  \noindent Standard algorithms for mean-payoff games~\cite{cominrizzi} with~\Cref{thm:mpgadgetcorrectness} gives an $\mathcal{O}(n^3mW)$ time algorithm (with $\abs{E} = m$) for solving fair mean-payoff games with threshold and an $\mathcal{O}(n^4mW)$ time algorithm for the optimal value problem.
  % We know that threshold value problem for mean-payoff games can be solved in time $\mathcal{O}(n^2mW)$~\cite{} where $\abs{V} = n$ and $\abs{E} = m$. Therefore, the above theorem implies that fair mean-payoff games can be solved in time $\mathcal{O}(n^4\cdot m \cdot W)$.
As regular mean-payoff games are determined, we also get the following from~\Cref{thm:mpgadgetcorrectness}:
\begin{theorem}
    Fair mean-payoff games are determined.
\end{theorem}

The rest of this section discusses the proof of~\Cref{thm:mpgadgetcorrectness} for 1-fair mean-payoff games in~\Cref{subsec:gadget1-correctness} and for 2-fair mean-payoff games in \Cref{subsec:gadget2-correctness}. The detailed proof of~\Cref{thm:mpgadgetcorrectness} for 2-fair mean-payoff games is provided in~\Cref{sec:app-gadget2-correctness}. Finally,~\Cref{subsec:strategy-meanpayoff} discusses the memory requirements of strategies in fair mean-payoff games resulting from the gadget-based reductions.

\subsection{Proof of~\Cref{thm:mpgadgetcorrectness} for 1-fair Mean-Payoff Games}\label{subsec:gadget1-correctness}

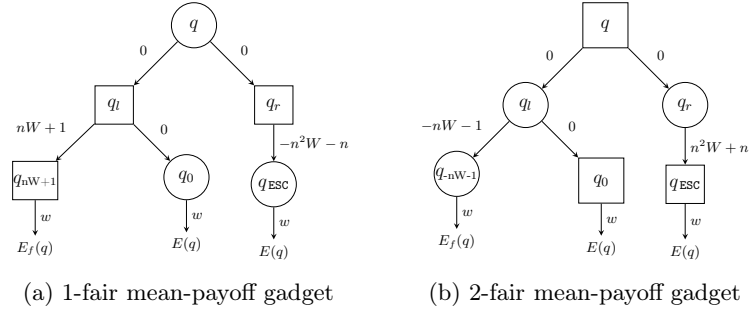
\begin{figure}[t]
  \centering
  \subfloat[%
  1-fair mean-payoff gadget\label{fig:Player1-gadget}]{
  \scalebox{0.6}{
    \begin{tikzpicture}[auto,node distance=2.5cm,semithick]
       \node[draw, circle, minimum size=1cm, inner sep=0pt, font=\large]
(q) {$q$};
       \node[rectangle,draw, minimum width=0.85cm, minimum
height=0.85cm, inner sep=0pt, font=\large] (ql) [below left of=q]
{$q_l$};
       \node[rectangle,draw, minimum width=0.85cm, minimum
height=0.85cm, inner sep=0pt, font=\large] (qr) [below right of=q]
{$q_r$};
       \node[draw, rectangle, minimum width=1cm, minimum height=0.85cm,
inner sep=0pt, font=\large] (qnw) [below left= 3.5em of ql]
{$q_\text{\scriptsize nW+1}$}; % color 1
       \node[draw, circle, minimum size=1cm, inner sep=0pt, font=\large]
(q0) [below right=3.5em of ql] {$q_0$}; % color 1
      \node[draw, circle, minimum size=1cm, inner sep=0pt, font=\large]
(qesc) [below=2.5em of qr] {$q_{\,\texttt{ESC}}$};
       \node (efq) [below of=qnw,yshift=1cm] {$E_f(q)$};
       \node (eql) [below of=q0,yshift=1cm] {$E(q)$};
       \node (eqr) [below of=qesc,yshift=1cm] {$E(q)$};

       \draw[->] (q) -- node[above left=0.1cm] {$0$} (ql);
       \draw[->] (q) -- node[above right=0.1cm] {$0$}(qr);
       \draw[->] (ql) -- node[above left=0.1cm] {$nW+1$} (qnw);
       \draw[->] (ql) -- node[above right=0.1cm] {$0$} (q0);
       \draw[->] (qr) -- node[right] {$-n^2W-n$} (qesc);
       \draw[->] (qnw) -- node[right] {$w$} (efq);
       \draw[->] (qesc) -- node[right] {$w$} (eqr);
        \draw[->] (q0) -- node[right] {$w$} (eql);
    \end{tikzpicture}
  }
  }
  \qquad
  \subfloat[%
 2-fair mean-payoff gadget\label{fig:Player2-gadget}]{\scalebox{0.6}{

  \begin{tikzpicture}[auto,node distance=2.5cm,semithick]
    \node[draw, rectangle, minimum width=0.85cm, minimum
height=0.85cm,minimum size=1cm, inner sep=0pt, font=\large] (q) {$q$};
    \node[draw,circle,minimum size=1cm,inner sep=0pt, font=\large] (ql)
[below left of=q] {$q_l$};
    \node[draw, circle,minimum size=1cm,inner sep=0pt, font=\large] (qr)
[below right of=q] {$q_r$};
    \node[draw, circle,minimum size=1cm,inner sep=0pt, font=\large]
(qnw) [below left=3.5em of ql] {$q_\text{\scriptsize -nW-1}$}; % color 1
    \node[draw, rectangle, minimum width=0.85cm, minimum
height=0.85cm,minimum size=1cm, inner sep=0pt, font=\large] (q0) [below
right=3.5em of ql] {$q_0$}; % color 1
    \node[draw, rectangle, minimum width=0.85cm, minimum height=0.85cm,
inner sep=0pt, font=\large] (qesc) [below=2.5em of qr] {$q_{\,\texttt{ESC}}$};
    \node (efq) [below of=qnw,yshift=1cm] {$E_f(q)$};
    \node (eql) [below of=q0,yshift=1cm] {$E(q)$};
    \node (eqr) [below of=qesc,yshift=1cm] {$E(q)$};

    \draw[->] (q) -- node[above left=0.1cm] {$0$} (ql);
    \draw[->] (q) -- node[above right=0.1cm] {$0$} (qr);
    \draw[->] (ql) -- node[above left=0.1cm] {$-nW-1$} (qnw);
    \draw[->] (ql) -- node[above right=0.1cm] {$0$} (q0);
    \draw[->] (qr) -- node[right] {$n^2W+n$} (qesc);
    \draw[->] (qnw) -- node[right] {$w$} (efq);
    \draw[->] (qesc) -- node[right] {$w$} (eqr);
    \draw[->] (q0) -- node[right] {$w$} (eql) ;
  \end{tikzpicture}

  }}
  \caption{Gadgets for converting fair mean-payoff games to
mean-payoff games. Edges are labeled with their
weights. The edges labeled with $w$ represent: an edge from a gadget node to a node $q' \in E(q)$ or
$E_f(q)$ carry the weight $w(q, q')$.}\label{fig:gadgets}
\end{figure}

\smallskip\noindent\textbf{Gadget Construction \& Intuition.}
Let $(G, E_f, \meanpayoff^f_0)$ be a 1-fair mean-payoff game. We
construct the \emph{gadget game} $G'$ by replacing every fair node in $q \in Q_f$ in $G$, with
its 3-step gadget presented in \cref{fig:Player1-gadget}.
That is, the incoming edges of $q$ is redirected to node $q$ at the root, and the outgoing edges on the third step lead to $E_f(q)$ and $E(q)$, the fair successors and successors of $q$, respectively.
The formal construction of $G'$ can be found in~\Cref{sec:app-formal-construction-of-1-mean-payoff-gadget}.

For notational convenience, for each $q \in Q_f$, we denote the leftmost
branch $q \to q_l \to q_\text{nW+1}$ of the gadget as
\emph{fair branch} $\fair{q}$, the middle branch $q \to q_{l}
\to q_{0}$ as the \emph{simulation branch} $\simul{q}$ and
the rightmost branch $q \to q_r \to q_{\ESC} $ as the
\emph{escape branch} $\escape{q}$. Intuitively, from a fair node $q$,
player~1 can escape to a part of the game that doesn't contain $q$ and is player~1 winning
(according to the current player~1 strategy) by taking the escape branch
and therefore paying a high negative payoff.
Since by this escape choice $q$ is guaranteed to be seen only finitely
often, the negative payoff associated with the escape edge does not change the winner of the game. 
Now assume there is no such winning escape
choice for player 1 from $q$. Then from $q$, all player 1 winning play visits $q$ infinitely often, and thus is forced to visit all of $q$'s outgoing fair edges infinitely often.
If taking one of its fair outgoing edges 
pushes the game into a player~2 winning region that does
not contain $q$
(according to some player~2 strategy), then in the gadget game player~2
can pay a high positive payoff and choose the fair branch with the correct successor to
escape to a player~2 winning region. So, in a way \enquote{\emph{fair branch is the escape branch for player~2}}.
However, player~2 can take this escape branch to win from $q$, \emph{only if} player~1 cannot take her escape branch to win.
The difference in the
amplitude of weights between the escape branches of the two players
(i.e. $\escape{q}$ and $\fair{q}$) stems from this hierarchy 
between the escape choices of the players. 
 %the order between the players choices: player 1 gets to chose its escape branch first, if it has a escape strategy. Only if player 1 does not have an escape strategy, then player 2 gets to chose the fair branch.
If neither player has an escape choice, then the middle
branch $\simul{q}$ faithfully simulates the fair mean-payoff game without adding any additional weights to the play via the gadget edges to alter the winner.

Using this intuition, we now prove both directions of \cref{thm:mpgadgetcorrectness} separately for 1-fair mean-payoff games.
% \begin{proof}[\Cref{thm:mpgadgetcorrectness} for 1-fair mean-payoff games]

\smallskip\noindent\textbf{1. Proof of~$\mathbf{(\Win_1(G') \, \cap \,Q \, \subseteq \,\Win_1(G))}$:} 
     Let $\sigma'$ be a positional player~1 strategy that wins $q$ in $G'$ for the regular mean-payoff objective. We construct a player~1 strategy $\sigma$ that wins $q$ in $G$.

    Consider the subgame arena $G'_{\sigma'}$ of $G'$.  All cycles in $G'_{\sigma'}$ are non-negative; else player~2 can force to eventually only visit a negative cycle and therefore construct a play in $G'_{\sigma'}$ with negative mean-payoff value.

    Using $\sigma'$, we construct $\sigma$ in $G$ as follows:
    \begin{enumerate}[label=\alph*.]
        \item\label{item:1-fair-zeroth} If $q' \not \in Q_f$, then set $\sigma$ to be positional at $q'$ with $\sigma(q') = \sigma'(q')$. %Note that this is well-defined since we do not use gadgets for nodes in $Q \setminus Q_f$.

        \item\label{item:1-fair-first} If $\sigma'(q') = q_r$ for some fair node $q' \in Q_f$, then set $\sigma$ to be positional at $q'$ with $\sigma(q') = \sigma'(q'_{\ESC})$. Intuitively this says if for a fair node $q'$, $\sigma'$ asks to choose the escape branch of the gadget, then $\sigma$ follows $\sigma'$ to choose that branch and moves to the corresponding `escape node' in $G$. 
        
        \item\label{item:1-fair-second} If $\sigma'(q') = q_l$ for some fair node $q' \in Q_f$, then we set $\sigma$ to be an infinite memory strategy at $q'$ that is played in rounds: consider an order on the fair successors of $q'$. In round $i \geq 0$, $\sigma$ chooses the `preferred successor' $\sigma'(q'_0)$ exactly $i$ times and then chooses the $i^\text{th}$ (modulo $|E_f(q')|$) fair successor. 
        
        % Formally, it keeps a counter $\counter$ for the rounds of the play, which is initialized at $1$; and a pointer $\mathbf{r}$ at the fair successors $E_f(p)$ of the node $p$. At round $\mathbf{i} = i$, $\sigma$ chooses the \enquote{preferred sucessor} $\sigma'(p_0)$ exactly $i$ times, and then takes the $\mathbf{r}^{th}$ fair successor once. Subsequently, it increments counter values of $\mathbf{i}$ and $\mathbf{r}$. When the value of $\mathbf{r}$ exceeds $|E_f(p)|$, it is reset to the first fair successor.
        % %local memory of size $\log(n^3W  +  n + |E_f(q)| + 2)$ make $s$ cycle through the set $\{s'(p_{0})\} \cup E_f(p)$ whenever it visits $p$, with a strategy that
        %spends substantially longer time at $p$s \enquote{preferred successor} $s'(p_{0})$ to collect enough positive weight, before it cycles
        %through its fair edges to fulfill the fairness condition. When $p$ is visited for the $k^{th}$ time, if $k \mod n^3W + n +2 \neq 0$, $s$ takes the edge $p \to s'(p_{0})$. 
        %If $k \mod n^3W  + n +2 = 0$, $s$ takes an edge in $E_f(p)$, starting from an arbitrary fair edge and cycling through all of them one-by-one whenever this condition holds.
    \end{enumerate}

    \noindent Now we prove that all $\tau \in \plays_\sigma(G)$ are winning for player~1 in $G$. In particular, we will show that $\tau$ is fair and $\meanVal(\tau) \geq 0$.
    
    \smallskip
    For a play $\tau = q^0 q^1 \ldots \in \plays_\sigma(G)$, we construct its \emph{extension} $\tau'$ in $G'$ inductively as follows: we start with $q^0$ and $i=0$ and with increasing $i$,
    \begin{itemize}
      \item If $q^i \not \in Q_f$, append $\tau'$ with $q^{i+1}$;
      \item If $\sigma(H\cdot q^i)$ is defined using~\Cref{item:1-fair-first} i.e., $\sigma'(q^i) = q^i_r$ then append $\tau'$ with $\escape{q^i} \to q^{i+1}$. In particular, we append $\tau'$ with $q^i_r \cdot q^i_{\ESC}\cdot q^{i+1}$.
      \item If $\sigma(H \cdot q^i)$ is defined using~\Cref{item:1-fair-second}, then
        \begin{itemize}
          \item if $q^{i+1} = \sigma'(q^i_{0})$, append $\tau'$ with $\simul{q^i}\to q^{i+1}$ i.e., with $q^i_l \cdot q^i_{0} \cdot q^{i+1}$ ; 
          \item else, append $\tau'$ with $\fair{q^{i}} \to q^{i+1}$, i.e. $q^i_l \cdot q^i_\text{nW+1} \cdot q^{i+1}$.
        \end{itemize} 
    \end{itemize}

    Clearly, $\tau'$ is a play in $G'_{\sigma'}$; hence it contains only non-negative cycles. Further, the restriction of $\tau'$ to the nodes $Q$ of $G$, denoted by $\tau'|_Q$, is exactly $\tau$. 

    Let a `tail' of $\tau$ denote an infinite suffix $\tail$ of $\tau$ that contains only the nodes and edges visited infinitely often in $\tau$; and let $\tail'$ be the extension of $\tau$ in $G'$. We denote the subgraph of $G'_{\sigma'}$ restricted to the nodes and edges in $\tail'$ by $G'_{\tail'}$.

    We first show that $\tau$ is fair. It is sufficient to show that $\tail$ does not contain fair nodes $p$ for which $\sigma(p)$ is defined via~\Cref{item:1-fair-first}, because for all the other fair nodes in $\tail$, $\tau$ visits all their fair successors infinitely often (\Cref{item:1-fair-second}).
     Let $p$ be a node for which $\sigma(p)$ is defined via~\Cref{item:1-fair-first}. Then $\sigma'(p) = p_r$. As $\sigma'$ is positional, only the escape branch of $p$ exists in $G'_{\sigma'}$, carrying the weight $-n^2W - n$. 
    Since $p$ is visited infinitely often, there is a (simple) cycle in $G'_{\sigma'}$ that passes through $p$, and thus sees the weight $-n^2W -n$. In order to maintain a non-negative weight, the cycle needs to contain $n$-many fair branches (carrying weight $nW+1$), which is not possible since there are at most $n-1$ many fair branches in $G'_{\sigma'}$. With this we conclude that $\tau$ is fair.
    
    Now we prove that $\meanVal(\tau) \geq 0$. First let us make an observation about the cycles in $G'_{\sigma'}$ that use only simulation branches. 

    \begin{observation}[sim-cycles] We call cycles $\mathbf{c'}$ in $G'_{\sigma'}$ that do not contain any fair or escape branches \emph{sim-cycles}. That is, $\mathbf{c'}$ contains only simulation branches. Being part of $G'_{\sigma'}$, $\mathbf{c}$ has non-negative weight.
    Since all the gadget branches in $\mathbf{c'}$ carry $0$ weight, the restriction $\mathbf{c} = \mathbf{c'}|_Q$ is a non-negative weight cycle in $G$.
    \end{observation}
    
    We saw that $\tail$ does not contain nodes defined using~\Cref{item:1-fair-first}. I.e., for all fair nodes $p$ in $\tail$, $\sigma$ is defined via~\Cref{item:1-fair-second}, and $G'_{\tail'}$ contains the branches $\fair{p}$ and/or $\simul{p}$. Thus, all cycles in $\tail'$ are either sim-cycles, or visit a fair branch. 
    
   If $\tail$ does not contain any fair nodes, $\tail = \tail'$, and $\meanVal(\tau) \geq 0$ easily follows. Now we assume $\tail$ contains at least one fair node, and use the following easy-to-verify remark to show $\meanVal(\tau) \geq 0$:
    
    \begin{remark}\label{lem:meanvalinfix}
      Given a play $\rho$ and $\epsilon >0$, if there exists a $k \in \mathbb{N}$ and a tail $\hat{\rho}$ of $\rho$ such that every $k$-length infix $\hat{\rho}_k$ of $\hat{\rho}$ satisfies $\meanVal(\hat{\rho}_k) \geq -\epsilon$, then $\meanVal(\rho) \geq -\epsilon$.
    \end{remark}
    Let $\epsilon > 0$ be arbitrary. Let $k\in \N$ such that $\frac{n^2W}{k} < \epsilon$. Since $\tail$ has at least one fair node, \Cref{item:1-fair-second} is triggered infinitely often to construct the strategy $\sigma$. Therefore, the round (in the sense of~\Cref{item:1-fair-second}) will grow unboundedly bigger. Thus, for any $k$, we can get a tail $\tail$ of $\tau$ that is in round $k$. In $\tail$, for every fair node, their edge to their preferred successor is taken at least $k$-many times between any two occurrences of their other (fair) outgoing edges. Then in $\tail'$, the simulation branches of fair nodes are visited at least $k$-many times in between two occurrences of fair branches.
    Therefore, in $\tail'$ each infix $\tail'_k$ of length $k$ visits the fair branch of each node at most once. Consequently, there are at most $n$-many simple cycles in $\tail'_k$ that visits a fair branch, and all its other simple cycles are sim-cycles. Now, total weight of $n$-many simple cycles that are not sim-cycles are $\geq -n^2W$, and since all the sim-cycles have non-negative weight, $\meanVal(\tau_k) \geq \frac{-n^2W}{k} \geq - \epsilon$. Using the fact from the previous paragraph, we get that $\meanVal(\tau) \geq -\epsilon$. This implies that for any $\epsilon > 0$, $\meanVal(\tau) \geq -\epsilon$. Hence, $\meanVal(\tau) \geq 0$.

    This concludes the proof of the first direction. %Now let us prove the other direction. In particular, we show $\Win_2(G') \, \cap \,Q \, \subseteq \,\Win_2(G)$. 

    \smallskip\noindent\textbf{2. Proof of~$\mathbf{(\Win_2(G') \, \cap \,Q \, \subseteq \,\Win_2(G))}$:}
     Analogous to the previous part, we consider a positional winning strategy $\pi'$ of player~2 in $G'$ and construct, this time, a \emph{positional} winning strategy $\pi$ for player~2 in $G$. Note that, none of the fair nodes in $G$ belong to player~2 and hence the construction of $\pi$ is quite straight forward: for all $q \in Q_2$, we set $\pi(q) = \pi'(q)$. We define $G'_{\pi'}$ as the subgame restricted to $\pi'$, as before. 

    For a play $\tau = q^0 q^1 \ldots \in \plays_\pi(G)$, we define its \emph{extension} $\tau'$ in $G'$ inductively as follows: we start from with $q^0$ and $i = 0$, and with increasing $i$,
    \begin{itemize}
        \item If $q^i \not \in Q_f$, append $\tau'$ with $q^{i+1}$; Otherwise,
        \item If $\simul{q^i}$ exists in $G'_{\pi'}$, append $\tau'$ with $q^i_l \cdot q^i_{0} \cdot q^{i+1}$;
        \item If $\simul{q^i}$ does not exist in $G'_{\pi'}$:
        \begin{itemize} 
            \item If $q^{i+1} = \pi'(q^i_\text{nW+1})$, append $\tau'$ with the fair branch i.e. $q^i_l \cdot q^i_\text{nW+1} \cdot q^{i+1}$, 
            \item Otherwise, append it with the escape branch, i.e. $q^i_r \cdot q^i_{\ESC} \cdot q^{i+1}$.
        \end{itemize}
    \end{itemize}
   % We define $\sigma$, $\sigma'$ and $G'_{\sigma'}$ as before. 
    We define the tail of $\tau$ as $\tail$ and its extension $\tail'$, analogous to the previous part. We denote the subgraph of $G'_{\pi'}$ restricted to the nodes and edges in $\tail'$ by $G'_{\tail'}$.

    Let $\tau \in \plays_\pi(G)$ be a play in $G$ that conforms with $\pi$. If $\tau$ is not fair, it is automatically won by player~2. Therefore, assume that $\tau$ is fair. We show that all the simple cycles in $\tail$ has negative weight. For this, we use the fact that its extension $\tau'$ is a play in $G'_{\pi'}$ and hence all the cycles in $G'_{\tail'}$ have negative weight. Since the length of simple cycles are bounded by $n$, this gives us $\meanVal(\tau) < 0$.

    \noindent The main argument is the absence of fair branches in $G'_{\tail'}$, introduced in~\Cref{claim:1-fair-absence}.
    
    \smallskip 
    \begin{claim}\label{claim:1-fair-absence}
      If $\tau$ is fair, there exist no fair branches $(\fair{})$ in $G'_{\tail'}$.
    \end{claim}

    \noindent Before proving the claim, we discuss how it implies that all simple cycles in $\tail$ are negative, resulting $\meanVal(\tau) < 0$. From the definition of extension, absence of fair branches in $G'_{\tail'}$ implies that either (i) for all $q \in Q_f$ only $\simul{q}$ exists in $G'_{\tail'}$, or (ii) for some $q \in Q_f$, only $\escape{q}$ exists in $\tail'$. By construction of $\tau'$, case (ii) above implies the existence of a fair node $q \in Q_f$ that appears infinitely often in $\tau$ but does not take all its fair outgoing edges, in particular $q \to \pi'(q_\text{nW+1}) \in E_f(q)$ infinitely often. This makes $\tau$ unfair, contradicting our assumption. In case of (i), since all simple cycles $\mathbf{c}'$ in $\tail'$ are negative and none of the gadget nodes contribute any weight due to only $\simul{}$ existing in $\tail'$, we conclude that the cycles $\mathbf{c} = \mathbf{c'}|_{Q}$ in $\tail$ have negative weight. This implies that $\meanVal(\tau) < 0$. 
    
    By the above-mentioned argument, proving~\Cref{claim:1-fair-absence} concludes the second direction of the proof.

     \smallskip
    \noindent \textit{Proof of~\Cref{claim:1-fair-absence}.} As discussed above, whenever $\escape{q}$ is in $G'_{\tail'}$ for a node $q$, $\fair{q}$ is also in it.
    Now we remove all the escape branches from $G'_{\tail'}$ and obtain the subgraph $S$. Note that $S$ has no dead-ends, since for all $q$ for which $\escape{q}$ got removed from $G'_{\tail'}$, $\fair{q}$ was also in $G'_{\tail'}$.
    
    For all nodes $q$ in $S$, either $\fair{q}$ or $\simul{q}$ is in $S$. As a subgraph of $G'_{\pi'}$, all cycles in $S$ have negative weight.
     It is easy to see that a fair branch cannot lie on a simple cycle in $S$, because fair branches carry weight $nW+1$ and the other edges come either from $\simul{}$ carrying weight $0$, or from $G$ carrying weight $\geq -W$. Hence, a simple cycle containing a fair branch would have positive weight. %but all cycles in $S$ have negative weight, as it is a subgraph of $G'_{\pi'}$.
    
    Now assume fair branches exist in $\tail'$, but not on cycles of $S$. Then all cycles in $S$ consist only of simulation branches. However, nodes with simulation branches have the same outgoing edges in 
    $G'_{\tail'}$ and in $S$. Since every node in $S$ is visited infinitely often in $\tau$, $\tau$ will eventually enter cycles of $S$ and never leave them, since these nodes have the same outgoing edges in $G'_{\tail'}$. This implies fair branches are not infinitely often visited in $\tau$, and concludes the proof of~\Cref{claim:1-fair-absence}.
    %Since $S$ has no dead-ends, it contains a cycle. 
    % Then $\tau$ will eventually visit a node on a cycle of $S$. Since this node will have the same outgoing edges in $S$ and $\tail'$, $\tau$ will never leave the cycles of $S$, which only have simulation branches. 
    % This implies that fair branches are not visited after some point on in $\tau$, and contradicts our assumption that fair branches exist in $\tail'$.
    \qed
    % This concludes the proof of the second direction and the theorem.
% \end{proof}

% \begin{corollary}
% \ISin{the proof implies the equivalence of these two definitions of determinacy, a.k.a. swappability}
% \end{corollary}

\subsection{Proof of~\Cref{thm:mpgadgetcorrectness} for  2-fair Mean-Payoff Games}\label{subsec:gadget2-correctness}

Observing page constrains, we provide the proof of~\Cref{thm:mpgadgetcorrectness} for 2-fair mean-payoff games
in \Cref{sec:app-gadget2-correctness} and only discuss its differences to \Cref{subsec:gadget1-correctness} here. %for 2-fair mean-payoff games.

Interestingly, the player~2 gadget for 2-fair games (\Cref{fig:Player2-gadget}) is exactly the dual of the player~1 gadget for 1-fair games (\Cref{fig:Player1-gadget}). This is surprising, as the winning objectives of player~1 and 2 are not exactly dual in these games. While player~1 in 1-fair games is expected to win the vertices with optimal value 0 (that is, the best value a player~1 strategy can achieve from these nodes is $0$, and getting $\epsilon$-close to the best possible value isn't sufficient to win); for player~2 to win a vertex in a 2-fair game, it is sufficient for him to get $\epsilon$-close to the best possible value he can achieve from this vertex. % That is if player~2 has a strategy to achieve the minimum mean-payoff value of $v$ from a node $q$, it wins the node via a strategy that achieves value $v + \epsilon$. However, if player~1 has a strategy to achieve the maximum mean-payoff value of $v$ from a node $q$; it does not necessarily have a strategy that wins $q$ with a value $v - \epsilon$. This occurrs only when $v = 0$.
As we will see in the next chapter, these differences in the objectives' behavior w.r.t.\ optimal values reflect drastically in the required strategy sizes. Namely, in 2-fair games finite strategies are sufficient whereas in 1-fair games infinite strategies are required to win. 

It is therefore surprising that 
this imbalance in objectives does not demand any changes in the gadget's structure for proving \cref{thm:mpgadgetcorrectness}   in the case of 2-fair mean-payoff games. The main difference w.r.t.\ the proof form \Cref{subsec:gadget1-correctness}
% The main difference between the proofs of the gadget for 2-fair mean-payoff games (\Cref{fig:Player2-gadget}) and the gadget for 1-fair mean-payoff games (\Cref{fig:Player1-gadget}) 
lies in the construction of winning strategies for the fair player. In particular, the proof for 1-fair games (in  \Cref{subsec:gadget1-correctness}) is slightly more complicated due to the required infinite strategy construction for player~2. On the other hand, the proof of 2-fair games (in \Cref{sec:app-gadget2-correctness}) reveals that a 
finite memory 
% local memory of size $\log(n^3W + n^2 + 2n + 1)$, and thus a global memory of $n^{\log(n^3W + n^2 + 2n + 1)}$ 
is sufficient for a winning player~2 strategy in 2-fair games. 

\subsection{Strategy Complexity for Fair Mean-Payoff Games}\label{subsec:strategy-meanpayoff}

We list an overview of results on strategy requirements for fair mean-payoff games. Most of these results follow from the proofs of~\Cref{thm:mpgadgetcorrectness}. 
We discussed in \Cref{sec:example} that player~1 may need infinite memory to achieve an optimal value in a 1-fair mean-payoff game. \Cref{lemma:mean-payoff-finite} shows that player~1 can reach $\epsilon$-close to the optimal value with finite memory strategies.

\begin{lemma}\label{lemma:mean-payoff-finite}
  Given a 1-fair mean-payoff game $(G,E_f,\meanpayoff^f)$, let the optimal value from some node $q_0$ is $v$. Then, for all $\epsilon > 0$, there exists a finite-memory strategy of player~$1$ that is winning from $q_0$ in fair mean-payoff game $(G,E_f,\meanpayoff(v-\epsilon))$.
\end{lemma} 

\begin{proof}
  W.l.o.g, we show this for $v=0$. In the proof of~\Cref{thm:mpgadgetcorrectness}, we construct an infinite memory strategy $\sigma$ for player~1 where at $i^\text{th}$ round, player~1 plays the preferred successor $i$ times and then plays a fair successor. Any $\sigma$ play has the following property: for any $\epsilon$, there exists a tail of the play such that the average weight of every $k$-length infix of the tail is at least $-\epsilon$. We fix this $k$ and modify the strategy construction in~\Cref{item:1-fair-second} to play the preferred successor $k$ times and then play a fair successor in every round. By~\Cref{lem:meanvalinfix}, this finite memory strategy ensures the mean-payoff value of any play is at least $-\epsilon$.\qed
\end{proof}

The above lemma entails that in a 1-fair mean-payoff game with threshold value $0$, player~1 has a finite memory strategy from a node $q \in \Win_1$ if the optimal value is strictly larger than $0$. On the other hand, in 2-fair mean-payoff games getting $\epsilon$-close to the best value achievable from a vertex is always sufficient to ensure winning for player~2. Intuitively, for this reason, player~2 has finite memory winning strategies in 2-fair mean-payoff games.

% For the nodes with optimal value $0$, the existence of a finite winning strategy is not guaranteed. On the other hand, in 2-fair mean-payoff games, a node $v \in \Win_2$ iff the best value a player~2 strategy can guarantee from $v$ is strictly less than $0$. That is, player~2 strategies do not have to achieve the best possible value for $v$ to win it, it is sufficient to get $\epsilon$-close to the best value to win a vertex. Intuitively, for this reason, player~2 has finite memory winning strategies in 2-fair mean-payoff games.

\begin{lemma}\label{lemma:mean-payoff-2-fair-finite}
  Given a 2-fair mean-payoff game $(G,E_f,\meanpayoff^f)$, for all $q \in \Win_2$, there exists a finite-memory strategy of player~$2$ that is winning from $q$.
\end{lemma} 

Finally, the proofs of~\Cref{thm:mpgadgetcorrectness} from \Cref{subsec:gadget1-correctness} (resp. \Cref{sec:app-gadget2-correctness}) reveal the existence of memoryless winning player~2 (resp. player~1) strategies in 1-fair (resp. 2-fair) mean-payoff games.
% Now for the other player, in the proof of, we have shown that player~2 has memoryless winning strategies in 1-fair mean-payoff games. Similarly, the gadget construction and the proof of~\Cref{thm:mpgadgetcorrectness} for 2-fair mean-payoff games show that player~1 has memoryless winning strategies.
%Hence, we have the following lemma.

\begin{lemma}\label{lemma:mean-payoff-2-fair-memoryless}
  For all fair mean-payoff games, memoryless strategies are sufficient for the player who does not own the fair nodes.
\end{lemma}

\section{Solving Fair Energy Games}\label{sec:solving-fair-energy}
In this section, we aim to extend our gadget-based approach to solve fair energy games. Intuitively, gadget-based approaches transfer the determinacy of regular (qualitative or quantitative) objectives to their fair variants, as demonstrated in \cite{irmak1} for fair parity games and by the proofs of~\Cref{thm:mpgadgetcorrectness} for fair mean-payoff games. Conceptually, these techniques build on finding winning strategies for either player in the fair game -- based on their winning strategies in the gadget game -- with similar infinite behavior. Therefore, before attempting a gadget-based approach for fair games, we need to ensure that these games are determined.

% In case of quantitative games, since $G'$ is determined, a winning player~2 strategy $\pi'$ wins against all initial credits $c$. If a projection player~2 strategy $\pi$ exists that wins from $\Win_2(G') \cap Q \subseteq \Win_2(G)$, it wins against all initial credits $c$ in $G$. This entails the determinacy of the fair game, together with $\Win_1(G') \cap Q \subseteq \Win_1(G)$ coming from the projection player~1 strategy $\sigma$.
% Therefore the determinacy of the fair quantitative game is required for gadgets based approaches (whose proofs depend on the construction of projection strategies) to be adaptable.

% We will start the section with a discussion on what it means to be determined for fair energy games in~\Cref{subsec:determinacyoffairenergy}. We will 
% see that a trivial definition of determinacy was adapted in previous work on $\omega$-regular games (not the determinacy we presented in the preliminaries), which allows easier solutions on these games. However, this trivial determinacy is not sufficient for gadget based approaches to be adapted, at least in sense the gadgets had been used before. For the gadget based approaches (in our sense) to be adaptable, we require the games to be determined in the stronger sense, as defined in the preliminaries.
In this section, we show that 2-fair energy games are \emph{not determined}, whereas 1-fair energy games are. While due to the above-mentioned reason we cannot construct gadgets for 2-fair energy games, we give simple algorithms to compute both $\Win_1$ and $\Win_2$. For 1-fair energy games, we observe that the gadget in~\Cref{fig:Player2-gadget} falls short, and we present new gadgets.

% We will give an alternative definition of the winning region for player~2 on $\omega$-regular energy games for our purposes, which we think is more natural from some aspect. We will show that 2-fair energy games are not determined with respect to this definition, where as 1-fair energy games are. We argue the lack of this determinacy prevents us from applying 
% a gadget-based technique to 2-fair energy games. On the other hand for 1-fair energy games we show that a gadget based technique works. For 2-fair energy games, we give a simple algorithm to compute $\Win_1$. 

%\input{sections/discussion-energy-determinacy.tex}

\subsection{Discussion on Determinacy of Energy Games}\label{subsec:discussion-short}
%  \noindent\textbf{Discussion on Determinacy of Energy Games.}
Whenever $\varphi_c$ is a Borel set, the following holds due to Borel determinacy~\cite{Martin75}:
 \begin{align*}
\forall \sigma \in \Sigma_1.\, \exists \pi \in \Sigma_2, \, \play_{\sigma, \pi}(G, q) \not \in \varphi_c \Leftrightarrow
\exists \pi \in \Sigma_2.\, \forall \sigma \in \Sigma_1,\,  \play_{\sigma, \pi}(G, q) \not \in \varphi_c 
\end{align*}
As $\energy_c$ and $\energy^f_c$ are Borel sets, the equation holds, showing that energy and fair energy games are determined under a fixed credit $c$. Furthermore, this equality combined with the negation of~\Cref{eq:determinacyLine1} defining $\Win_1$, yields the following formulation of $Q \setminus \Win_1$:
\begin{equation}
  q \not \in \Win_1 \iff 
  \forall c \in \mathbb{N}.  \,\exists \pi \in \Sigma_2.\, \forall \sigma \in \Sigma_1, \, \play_{\sigma, \pi}(G, q) \not \in \varphi_c \label{eq:energyDeterminacyPlayer2-triv-equivalent}
\end{equation}
Again, this formulation holds for both energy and fair energy games. Clearly it follows that a (fair) energy game is determined if and only if~\Cref{eq:energyDeterminacyPlayer2} is equivalent to~\Cref{eq:energyDeterminacyPlayer2-triv-equivalent}. In fact, if we can restrict the quantification over $c$ in~\Cref{eq:energyDeterminacyPlayer2-triv-equivalent} to a finite set, then we can swap $\forall c \in \mathbb{N}$ and $\exists \pi \in \Sigma_2$, which yields the desired equivalence, showing the game is determined. 
This is indeed the case in energy games: Whenever player~1 wins from a node, it wins with initial credit $nW$. Dually, whenever player~2 wins w.r.t. initial credit $nW$, it wins against every initial credit. That is, the quantification over the $c$ can be restricted to $[0, nW]$. %We will show in \Cref{sec:2-fair-energy} that this is not the case for 2-fair energy games.

Using this trick, we show in~\Cref{sec:1-fair-energy} that 1-fair energy games are determined. However, we demonstrate in the next section that this does not hold for 2-fair energy games and that these games are not determined.

\subsection{2-fair Energy Games}\label{sec:2-fair-energy}
We start by showing that 2-fair energy games are not determined.
\begin{theorem}
  2-fair energy games are not determined.
\end{theorem}

  \begin{proof}Recall that energy games are determined if and only if~\Cref{eq:energyDeterminacyPlayer2} is equivalent to~\Cref{eq:energyDeterminacyPlayer2-triv-equivalent}. We provide a counterexample to this equivalence by the game graph in \cref{fig:2-fair-energy-game}.
  Any fair strategy $\pi$ for player~2 takes the edge $(q,q')$ after finitely many steps, say after $c$ steps. Then, the unique strategy $\sigma$ for player~1, which takes the self-loop on $q'$, wins the play $\play_{\sigma}(G,q)$ with respect to the initial credit $c$.  Thus, player~$2$ has no winning strategy. 
  However, player~$1$ has no winning strategy either since for any initial credit $c$, there exists a player~$2$ strategy that ensures the energy of every play drops below $0$. Namely, any strategy that takes the self-loop of $q$ more than $c$ times achieves this.\qed
\end{proof}

\begin{wrapfigure}[6]{r}{0.3\textwidth}
  % \begin{figure}
  \centering
  \vspace{-1cm}
  \begin{tikzpicture}
    \node[player1] (q) {$q$};
    \node[player0] (qprime) [right of=q] {$q'$};
  
    \path
      (q) edge [loop above] node {-1} (q)
          edge [dashed] node {0} (qprime)
      (qprime) edge [loop above] node {0} (qprime);
  \end{tikzpicture}
 
  \caption{}~\label{fig:2-fair-energy-game}
  \end{wrapfigure}

  %The example is given in \Cref{fig:2-fair-energy-game}. 
  %Similar to energy parity case, the game is not determined due to a lack of a player~2 strategy that wins against all initial credits.

This occurs as player~2 can delay taking his fair edges enough to violate any given $c$ while still satisfying fairness condition in the suffix. Hence, if player~$2$ can force a negative cycle, then he can use it to violate any initial credit. In contrast, player~$1$'s objective remains the same as in the regular energy game: she wins from a node iff she can prevent player~$2$ from forcing a negative cycle from that node. This leads to the following result, which is proven in \Cref{sec:app-proofs-2-fair-energy}.
\begin{theorem}\label{thm:2-fairEnergyWin1}
  Given a 2-fair energy game $(G,E_f,\energy)$, the set of winning nodes for player~$1$ is the same as the set of winning nodes for player~$1$ in the (regular) energy game $(G,\energy)$.
\end{theorem}
% 
% 
% As a corollary of~\Cref{thm:2-fairEnergyWin1}, we obtain the following result.
% 
\begin{lemma}\label{thm:2-fairEnergMemoryless}
  In 2-fair energy games, player~1 has memoryless winning strategies. 
\end{lemma}

We note that \Cref{thm:2-fairEnergMemoryless} is a corollary of~\Cref{thm:2-fairEnergyWin1}. 
As 2-fair energy games are not determined, $\Win_2 \neq Q \setminus \Win_1$.
We present a simple reduction to compute $\Win_2$ and resulting memory requirements of strategies as a corollary of its proof.

\begin{lemma}\label{lemma:2-fairEnergyWin2}
  Given a 2-fair energy game $(G,E_f,\energy)$, the set of winning nodes for player~$2$ is the same as the set of winning nodes for player~$2$ in the 2-fair mean-payoff game $(G,E_f, \meanpayoff^f_0)$.
\end{lemma}
 
\begin{proof}
  Take a player~2 strategy $\pi$ that wins from $q$ in the 2-fair mean-payoff game. Then, for any $\tau \in \plays_{\pi}(G, q)$, $\meanVal(\tau) = - \epsilon$ for some positive $\epsilon$. That is, there exists an infinite subsequence $I$ of $\mathbb{N}$ such that for all $i \in I$, $w(\tau[0;i]) < - i \cdot \epsilon$. Thus, the energy of $\tau$ drops unboundedly.
  For the opposite direction, due to the determinacy of mean-payoff games and~\Cref{lemma:mean-payoff-2-fair-memoryless}, it suffices to show that for each positional player~1 strategy $\sigma$, there exists a player~2 strategy $\pi$ such that
 $\tau = \play_{\sigma, \pi}(G, q)$ has negative limit average weight. Let $\sigma$ be a positional player~1 strategy and $\pi$ a player~2 strategy that wins $q$ in the 2-fair energy game. Consider $\tau = \pr \cdot \tail$, where $\tail$ is a tail of $\tau$ consisting of infinitely often visited nodes and edges of $\tau$. Since the energy of $\tail$ drops unboundedly, there exists a negative cycle $\mathbf{c}$ in $\tail$ that visits all nodes and (fair) edges of $\tail$. Let $\pi_{\sigma}$ be a player~2 strategy that mimics $\pi$ until $\mathbf{c}$ is seen, and then repeats $\mathbf{c}$. Then $\rho = \play_{\sigma, \pi_{\sigma}}(G, q)$ has a tail $\mathbf{c}^\omega$, and since $w(\mathbf{c})< 0$, we have $\meanVal(\rho) < 0$. 
  %Now take a finite memory (\Cref{thm:2-fairEnergyFiniteMemory}) winning player~2 strategy $\pi$ in the 2-fair energy game, and a memoryless player~1 strategy $\rho$ (due to~\Cref{lemma:mean-payoff-2-fair-memoryless}, we can restrict ourselves to such strategies). Due to~\Cref{eq:energyDeterminacyPlayer2-2}, for each $\tau \in \plays_{\pi}(G)$, the energy of $\tau$ drops unboundedly. 
  %Let $\tau = \play_{\rho, \pi}$. Since $\rho$ and $\pi$ are both finite memory, $\tau$ has a tail of the form $\uu^\omega$ where $\uu$ is a cycle in $G$. Clearly, $w(\uu)$ is negative. Therefore, $\meanVal(\tau)$ is also negative.  
\qed\end{proof}

% The next lemma follows as a corollary of the proof of~\Cref{lemma:2-fairEnergyWin2}.
 \begin{lemma}\label{lemma:2-fairEnergyFiniteMemory}
   Player 2 has finite memory winning strategies in 2-fair energy games.
\end{lemma}

\subsection{1-fair Energy Games}\label{sec:1-fair-energy}
In this section, we prove that 1-fair energy games are determined and using that introduce a gadget to reduce 1-fair energy games to regular energy games.

\noindent The determinacy of 1-fair energy games does not follow from prior work on $\omega$-regular energy games~\cite{EnergyMuCalculus2020} as the definition of $\Win_2$ there uses~\Cref{eq:energyDeterminacyPlayer2-triv-equivalent} instead of~\Cref{eq:energyDeterminacyPlayer2}.
We establish determinacy by showing that, as in energy games, one can restrict the quantification over $c \in \mathbb{N}$ to a finite set (namely to $c \in [0, nW]$) in~\Cref{eq:energyDeterminacyPlayer2-triv-equivalent}.
That is, if there exists a player~2 strategy $\pi$ that satisfies $\forall \sigma \in \Sigma_1, \play_{\sigma, \pi}(G, q) \not \in \energy^f_{nW}$, then there exists a winning player~2 strategy $\pi^{\ext}$ (winning against all $c$) as the energy level of plays conforming to this strategy $\pi^{\ext}$ drops unboundedly. It then follows from the discussion in~\Cref{subsec:discussion-short} that 1-fair energy games are determined. The proof of~\Cref{thm:determinacy1-fairenergy} can be found in~\Cref{sec:app-proofs-P1-fair-energy}.

\begin{theorem}\label{thm:determinacy1-fairenergy}
1-fair energy games are determined.
\end{theorem}

%  We show that as in regular energy games, one can restrict the quantification over $c \in \mathbb{N}$ to a finite set (namely to $c \in [0, nW]$) in \Cref{eq:energyDeterminacyPlayer2-triv-equivalent}. That is, we show that if there exists a player~2 strategy $\sigma$ that satisfies $\forall \pi \in \Sigma_1, \play_{\sigma, \pi}(G, q) \not \in \energy^f_{nW}$, then there exists a winning player~2 strategy $\sigma^{\ext}$. Note that for $\sigma^{\ext}$, the energy level of conforming plays drop below any integer. As discussed in~\Cref{subsec:discussion-short}, this gives us determinacy of 1-fair energy games.

%  \smallskip
\noindent Next we will show that the gadget in~\Cref{fig:1-fair-energy-gadget} turns 1-fair energy games into regular energy games. 
The difference of this gadget from~\Cref{fig:Player1-gadget} mostly lies in the different treatment of the zero-weight-cycles. We first discuss briefly why such cycles are important, and then introduce this gadget.

\smallskip\noindent \textbf{The importance of 0-cycles in 1-fair energy games.}
It is known that the winning regions of a player in energy games and mean-payoff games (with objective $\meanpayoff_0$) coincide~\cite{FasterAlgorithmsforMeanPayoffBCDGR11}.
% 
% a finite-memory strategy is winning for player~1 in energy games if and only if it is winning for player~1 in mean-payoff games (with objective $\meanpayoff_0$)~\cite{FasterAlgorithmsforMeanPayoffBCDGR11}.
% 
% In both energy and mean-payoff games, a player 1 strategy is winning if and only if all the cycles conforming with the strategy have non-negative weight. 
This is not true for the 1-fair variants of these games, as player~1 may need infinite memory to win in 1-fair mean-payoff games (see \Cref{sec:example,subsec:gadget1-correctness}), but not in 1-fair energy games.
% a strategy can make increasing use of a non-negative weight cycle to win in the 1-fair mean-payoff game (recall the infinite memory optimal value player 1 strategies in~\Cref{sec:example,subsec:gadget1-correctness}), but it may not be able to do so in the 1-fair energy game.
% 1-fair mean-payoff might need infinite memory strategies to win whereas 1-fair energy games do not.
% It has been observed before\cite{FasterAlgorithmsforMeanPayoffBCDGR11} that energy and mean-payoff games have the same winning regions on graphs without $0$-cycles. 
However, the equivalence of winning regions still holds in fair games without $0$-cycles.
% 1-fair energy and 1-fair mean-payoff games. 
Intuitively, if there are no $0$-cycles, for a winning player~1 (player~2) strategy $\sigma$ ($\pi$) in the fair energy game, the energy of each play $\tau$ that conforms with $\sigma$ ($\pi$) grows (drops) unboundedly, which makes $\tau$ also winning in the fair mean-payoff game. 
For instance, in a 1-fair energy game with a player~1 node with a single fair edge with $0$ weight is winning for player~1; but if we add another fair edge to this node with weight $-1$, it becomes losing for player~1. However, when we consider the graph under the mean-payoff objective, both nodes remain in $\Win_1$.

\begin{figure}
% \begin{wrapfigure}{r}{0.35\textwidth}
  % \vskip -0.82cm
   \raggedright
   \centering 
   \scalebox{0.8}{
    \begin{tikzpicture}[auto,node distance=1.5cm,semithick]
      \node[draw, circle, minimum size=1cm, inner sep=0pt, font=\large] (q) {$q$}; 
      \node[rectangle,draw, minimum width=0.85cm, minimum height=0.85cm, inner sep=0pt, font=\large] (ql) [below left=4em of q]{$q_l$};
      \node[rectangle,draw, minimum width=0.85cm, minimum height=0.85cm, inner sep=0pt, font=\large] (qr) [right=9em of ql]{$q_r$};
      \node[draw, rectangle, minimum width=1cm, minimum height=0.85cm,inner sep=0pt, font=\large] (qnw) [below left= 1.5em of ql]{$q_\text{\scriptsize nW+1}$}; % color 1
      \node[draw, circle, minimum size=1cm, inner sep=0pt, font=\large](qval) [below right=1.5em of ql] {$q_\text{val}$}; % color 1
     \node[draw, circle, minimum size=1cm, inner sep=0pt, font=\large](qpos) [below left=1.5em of qval] {$q_\text{pos}$}; % color 1
     \node[draw, rectangle, minimum size=1cm, inner sep=0pt, font=\large](q0) [right=4em of qpos] {$q_0$}; % color 1
     \node[draw, circle, minimum size=1cm, inner sep=0pt, font=\large](qesc) [right=1.45em of q0] {$q_{\,\texttt{ESC}}$};
     \node (efq) [below=7em of qnw,yshift=0.8cm] {$E_f(q)$};
     \node (eqm) [right=1em of efq] {$E(q)$};
     \node (efqm) [right=3.6em of eqm] {$E_f(q)$};
     \node (eqr) [right=1.45em of efqm] {$E(q)$};
 
      \draw[->] (q) -- node[above left=0.1cm] {$0$} (ql);
      \draw[->] (q) -- node[above right=0.1cm] {$0$}(qr);
      \draw[->] (ql) -- node[above left=0.1cm] {$nW+1$} (qnw);
      \draw[->] (ql) -- node[above right=0.1cm] {$0$} (qval);
      \draw[->] (qval) -- node[above left=0.1cm] {$- \varepsilon$} (qpos);
      \draw[->] (qval) -- node[above right=0.1cm] {$0$} (q0);
      \draw[->] (qr) -- node[right] {$-n^2W-n$} (qesc);
      \draw[->] (qnw) -- node[right] {$w$} (efq);
      \draw[->] (qesc) -- node[right] {$w$} (eqr);
       \draw[->] (q0) -- node[right] {$w$} (efqm);
       \draw[->] (qpos) -- node[right] {$w$} (eqm);
   \end{tikzpicture}
    }
    \caption{1-fair energy gadget where $\epsilon = 1/(n+1)$}~ \label{fig:1-fair-energy-gadget}
    \vskip -0.6cm
  % \end{wrapfigure}
  \end{figure}
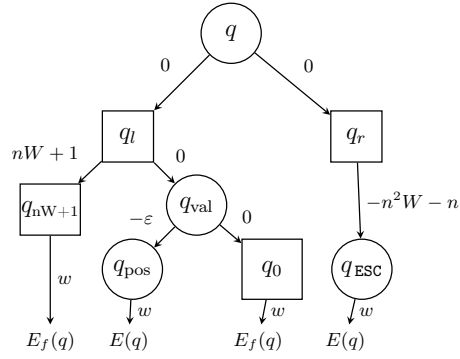
Therefore, a 1-fair energy gadget should distinguish between two kinds of nodes, A and B, both of which achieve mean-payoff value $0$ but
node A (in $\Win_1$) achieves value $0$ due to a strategy that visits only $0$-cycles while node B (in $\Win_2$) achieves it due to an infinite memory strategy that sees a $0$-cycle more and more often. 
%In particular, in a 1-fair energy game a player~1 node with a single fair edge with weight zero is winning for player~1\RIT{you mean this can be a case?}; however, if we add another fair edge to this node with weight $-1$, this node becomes losing for player~1. However, if we view this graph with the mean-payoff objective, both nodes are in $\Win_1$.
%Therefore a 1-fair energy game gadget should distinguish two kinds of nodes, A and B, that achieve a mean-payoff value $0$.
%Node A (in $\Win_1$) achieves value $0$ due to a strategy that sees only $0$-cycles, and node B (in $\Win_2$) achieves it due to an infinite memory strategy that sees a $0$ cycle more and more often.\RIT{The last two lines also appear exactly same as the beginning of the paragraph.}
 The middle branch of the gadget in~\Cref{fig:1-fair-energy-gadget} serves this purpose. The gadget matches~\Cref{fig:Player1-gadget} on the leftmost ($\fair{}$) and rightmost ($\escape{}$) branches, which still serve as the escape branches for player~1 and player~2, as in~\Cref{subsec:gadget1-correctness}.
The middle branch, called the \emph{value branch} $\valu{q}$ distinguishes nodes of type A and B. Intuitively, if both players do not have escape strategies from $q$, then the value branch is taken. If player~1 can visit a positive cycle from $q$, then she wins from $q$ since intuitively she can take this cycle enough times before visiting its fair edges. In this case, she chooses the left successor $q_\texttt{pos}$ of $q_\texttt{val}$, which is called the \emph{positive value branch} $\pos{q}$ and visits the weight $- \epsilon$. The game is won by player~1 despite the weight $-\epsilon$, due to the positive cycle. On the other hand, if player~1 is winning from $q$ but not by visiting a positive cycle, then she is winning as all the live edges lead to $0$-cycles (Node A). In this case, player~1 takes the right successor $q_\texttt{0}$ of $q_\texttt{val}$, which is called the \emph{zero value branch} $\zero{q}$, and wins. If neither of these cases holds for a node $q$ (Node B), player~2 wins from both successors of the value branch.

\begin{theorem}\label{thm:energygadgetcorrectness}
  Let $(G, E_f, \meanpayoff^f_0)$ be a 1-fair energy game where $G = (Q, E, w)$, $w : E \to [-W, W]$, and $|Q|=n$. Then there exists an energy game $G'= (Q',E',w')$ where $Q' \supseteq Q$, $|Q'|\leq 8n$ and $w' : E' \to [-W', W']$ with $W' = (n^2W+n)(n+1)$ such that
  $\Win_i(G) = \Win_i(G') \cap Q$ for $i \in \{1, 2\}$.
\end{theorem}

We obtain the game $G'$ in \Cref{thm:energygadgetcorrectness} by replacing each fair node in the 1-fair energy game $G$ via the gadget in~\Cref{fig:1-fair-energy-gadget}, as in~\Cref{subsec:gadget1-correctness}. We note that the size of the game given in
\Cref{thm:energygadgetcorrectness} is w.r.t. the equivalent energy game on integer weights, obtained by multiplying every weight in the game with $1 / \epsilon = n+1$. 
The state-of-the-art algorithm for energy games~\cite{FasterAlgorithmsforMeanPayoffBCDGR11} and~\Cref{thm:energygadgetcorrectness} imply that 1-fair energy games can be solved in time $O(n^4mw)$. Furthermore, the proof reveals that player~1 has finite memory strategies (a local memory of $log(n^3W + n)$ each node) and player~2 has positional strategies. 

Unlike fair mean-payoff games, the determinacy of 1-fair energy games \emph{does not} follow from~\Cref{thm:energygadgetcorrectness}. Instead, the proof hinges on the determinacy result itself, as detailed in~\Cref{sec:app-proofs-P1-fair-energy}. Below, we outline the key ideas.

Recall that in the first direction of the proof of~\Cref{thm:mpgadgetcorrectness}, we derive a player~1 strategy $\sigma$ in $G$ from a positional winning strategy $\sigma'$ in $G'$, reasoning about the tails $\tail$ of plays $\tau \in \plays_{\sigma}(G)$. We use the fact that in a mean-payoff (and similarly, in an energy) game, all cycles in $G'_{\sigma'}$ are non-negative. This gives us that the cycles $\mathbf{c'}$ using only simulation branches in $G'_{\sigma'}$ have projections $\mathbf{c} = \mathbf{c'}|_G$ with non-negative weight. In the 1-fair mean-payoff case, with infinite memory we let $\sigma$ take these cycles increasingly often to ensure that the projection play $\tau = \tau'|_Q$ has non-negative mean-payoff. 

Now, in 1-fair energy games, the cycles $\mathbf{c'}$ in $G'_{\sigma'}$ that visit a positive-value branch have \emph{positive} weight projections $\mathbf{c} = \mathbf{c'}|_Q$, due to the $-\epsilon$ weights. We refer to the choice of the positive-value branch, $\sigma'(q_{\texttt{pos}})$ as the `preferred successor' of $q$ that induces positive weight cycles in $G$.
Thus, player~1 can adopt a strategy $\sigma$ in $G$ that repeatedly selects the preferred successors at fair nodes to accumulate sufficient positive weight before traversing the fair edges to satisfy fairness. 
 As a result, the weight of every tail $\tail$ of all plays $\tau \in \plays_{\sigma}(G)$ 
 is monotonically increasing, and bounded below by some $-d$. However, unlike mean-payoff games, energy games require reasoning about prefixes of plays. Specifically, we need to ensure that the weight of every prefix of each play conforming to $\sigma$ is bounded below by $-c$ for some $c \in \N$. This is challenging, as the prefix lengths are unbounded, making the determinacy result crucial.
  %  However, as opposed to mean-payoff games, in energy games it is not sufficient to reason about the tail of plays, but we also need to reason about the prefixes. In particular, we need to show that the weight of every prefix of every play conforming with $\sigma$ is bounded below by $-c$ for some $c \in \mathbb{N}$. But this is hard to do, since the prefix of $\tau$ does not have a bounded length. This is where the determinacy result comes to the rescue.

 As these games are determined, player~2 wins from a node $q$ if and only if it has a winning strategy $\pi$ that wins against all initial credits $c$ and all player~1 strategies (\Cref{eq:energyDeterminacyPlayer2}).
Therefore, the weight of all plays in $\plays_\pi(G)$ drops unboundedly. 
Dually, a player~1 strategy $\lambda$ is winning if all plays in $\plays_\lambda(G)$ are fair with weights bounded below (not necessarily by the same $c$).
Then for all $\tau = \uu \cdot \tail$, the length of the finite prefix $\uu$, together with the lower bound $-d$ of the tail $\tail$ naturally yields a lower bound of $-d -W\cdot |\uu|$ for the weight of $\tau$.

The proof of the second direction is similar to the case of 1-fair mean-payoff, where we construct a positional winning player 2 strategy $\pi$ in $G$ based on a positional winning strategy $\pi'$ in $G'$.  
This shows that memoryless strategies suffice for player 2. We give an alternative, simpler proof of this in~\Cref{sec:app-proofs-P1-fair-energy}. 
\begin{theorem}\label{thm:1-fairenergypositionality}
  In 1-fair energy games, player~1 has finite memory and player~2 has memoryless winning strategies. 
\end{theorem}

\section{Conclusion}

In this work, we study the complexity of solving mean-payoff and energy games under strong transition fairness assumptions. We show that when combined with quantitative goals, fairness comes for free in the sense that the complexity of solving these games do not become computationally expensive. We discussed the determinacy and strategy complexity of these games and provided gadget-based techniques to solve them. 

A possible future direction is to study the complexity of solving the doubly fair counterpart of these games, where both players are restricted by fairness constraints on their moves, similar to the extension of fair $\omega$-regular games in~\cite{irmak1}.

In the $\omega$-regular setting, a slight extension of strong transition fairness called \emph{group fairness} is sufficient to make the problem $\NP$-hard, even when combined with the reachability objective~\cite{irmak3}. Group fairness is defined via tuples $(G_1, E_1), \ldots ,(G_k, E_k)$ where for each $i \in [1, k]$, $G_i$ is a subset of player 2 nodes, and $E_i$ is a subset of edges such that the source node of each edge $e \in E_i$ belongs to $G_i$. 
A play $\tau$ is called \emph{group fair} if, for each $v \in V$ and $i \in [1, k]$, whenever $v \in \inf(\tau) \cap G_i$, there exists an $e \in E_i$ that is taken infinitely often in $\tau$. 

In a group fair game with an $\omega$-regular objective $\varphi$ and tuples $(G_1, E_1), \ldots$ $, (G_k, E_k)$ defining the group fairness, player 1 wins a play $\tau$ if $\tau$ satisfies $\varphi$ or \emph{is not group fair}.
Being a subclass of Streett fairness (or, strong fairness), group fairness already unravels the full complexity of Streett fairness\footnote{Note that in a game where fairness is inflicted on player 2 (as in the group fairness definition), the fairness condition is a subclass of Rabin objective rather than Streett. Thus, the derived complexity lower bound is $\NP$ rather than $\coNP$.}. Since reachability can be expressed as a mean-payoff/energy objective, this $\NP$-hardness result carries over to 2-group-fair\footnote{By this, we indicate a game with group fairness where all $G_i$ are player 2 nodes, as in the original definition of group fairness.} mean-payoff games, making these games $\NP$-complete. 2-group-fair energy games again suffer from a lack of determinacy, since they subsume 2-fair energy games.
Similarly, as the dual of reachability, the safety objective can be expressed as a mean-payoff/energy objective. This allows us  to derive a $\coNP$ lower bound for 1-group-fair\footnote{By this, we indicate a game with group fairness where all $G_i$ are player 1 nodes. player 1 wins a play $\tau$ of a 1-group-fair game with objective $\varphi$ if $\tau$ satisfies $\varphi$ and \emph{is group fair}.} mean-payoff and 1-group-fair energy games, making these games $\coNP$-complete. 
However, the exploration of more general notions of fairness -- such as other subclasses of Streett fairness -- that could provide better complexity bounds or lead to elegant solution algorithms for qualitative or quantitative games remains a significant open area of research.

\label{beforebibliography}
\newoutputstream{pages}
\openoutputfile{main.pg}{pages}
\addtostream{pages}{\getpagerefnumber{beforebibliography}}
\closeoutputstream{pages}
\bibliographystyle{splncs04}
\bibliography{main}

\pagebreak

\newpage
\appendix

\section{Formal construction of the player~1 gadget in~\Cref{thm:mpgadgetcorrectness}}\label{sec:app-formal-construction-of-1-mean-payoff-gadget}
Formally, the `gadget game' $G'= (Q',E',w')$ with $Q' = Q'_1 \uplus
Q'_2$ is defined as follows:
\begin{itemize}
  \item For every fair node $q \in Q_f$, we introduce two new player~1
nodes $q_{0}, q_{\ESC}$ and three new player~2 nodes $q_l, q_r,
q_{\text{nW+1}}$ in $G'$.
    \begin{align*}
        Q'_1 = &Q_1 \cup \{q_{0}, q_{\ESC}\, \mid \, q \in Q_f\}\\
        Q'_2 = &Q_2 \cup \{q_l, q_r, q_\text{nW+1}\, \mid \, q \in Q_f\}
    \end{align*}
   \item With the introduction of the new nodes in the gadgets, we
introduce new edges as follows:
   \begin{align*}
    E' \,= &\,\, \{\,qE  \,\mid\, q \not \in Q_f\,\}\\
    &\,\cup \,\{\, (q, q_l),\, (q, q_r), \,(q_l, q_\text{nW+1}),\, (q_l,
q_{0}),\, (q_r, q_{\ESC}) \,\mid \, q \in Q_f\,\} \\
    &\,\cup \,\{\, (q_\text{nW + 1}, q') \mid q' \in E_f(q)\} \cup
\{(q_{0}, q'),\, (q_{\ESC}, q')\mid q' \in E(q)\}
   \end{align*}
   \item We define the new weight function $w' : E' \to [-W',W']$ with
$W' = n^2W+n$ where $n=|Q|$ as follows: for every $q \not \in Q_f,  w'(e) = w(e)$, for all $e \in qE$, and for
every $q \in Q_f$,
    \begin{align*}
      &w'(q,q_l) = w'(q,q_r) =w'(q_l,q_0) = 0,\\
      &w'(q_l, q_{\text{nW+1}}) = nW + 1,\\
      &w'(q_r, q_{\ESC}) = -n^2W - n,\\
      &w'(q_\text{nW+1}, q') = w(q, q') \text{ for } q' \in E_f(q),\\
      &w'(q_{0}, q') = w(q, q') \text{ for } q' \in E(q),\\
      &w'(q_{\ESC}, q') = w(q, q') \text{ for } q' \in E(q)\\
    \end{align*}
\end{itemize}

   \section{Proof of~\cref{thm:mpgadgetcorrectness} for 2-fair mean-payoff games}\label{sec:app-gadget2-correctness}

   In this section, we will prove~\cref{thm:mpgadgetcorrectness} for 2-fair mean-payoff games.
   We will do so via providing the correctness proof of the gadget in~\cref{fig:Player2-gadget}, parallel to the proof in~\cref{subsec:gadget1-correctness}.
   
   Although the proof is similar to that of 1-fair mean-payoff games, there is a substantial difference in strategy sizes. In 2-fair mean payoff games, the size of player~2 winning strategies are bounded by a local memory of $\log(n^3+ n^2+2n+2)$, as opposed to 1-fair mean payoff games that has possibly infinite memory winning strategies for player~1. 
   The reader curious about the reason of this difference in strategy sizes can jump to the second direction of the proof, and visit the definition of the winning player~2 strategy $\pi$, focusing on the bullet point c..

   We start by giving a formal definition of the gadget game $G'$ obtained by replacing each fair node in a 2-fair game with the 3-step gadget given in~\cref{fig:Player2-gadget}.

   Formally, the `gadget game' $G'= (Q',E',w')$ with $Q' = Q'_1 \uplus
   Q'_2$ is defined as follows:
   \begin{itemize}
     \item For every fair node $q \in Q_f$, we introduce two new player~1
   nodes $q_{0}, q_{\ESC}$ and three new player~2 nodes $q_l, q_r,
   q_{\text{nW+1}}$ in $G'$.
       \begin{align*}
        Q'_1 = &Q_1 \cup \{q_l, q_r, q_\text{-nW-1}\, \mid \, q \in Q_f\}\\
        Q'_2 = &Q_2 \cup \{q_{0}, q_{\ESC}\, \mid \, q \in Q_f\}
       \end{align*}
      \item With the introduction of the new nodes in the gadgets, we
   introduce new edges as follows:
   \begin{align*}
    E' \,= &\,\, \{\,qE  \,\mid\, q \not \in Q_f\,\}\\
    &\, \cup \,\{\, (q, q_l),\, (q, q_r), \,(q_l, q_\text{-nW-1}),\, (q_l, q_{0}),\, (q_r, q_{\ESC}) \,\mid \, q \in Q_f\,\} \\
    &\,\cup \,\{\, (q_\text{-nW-1}, p^f), \,(q_{0}, p),\, (q_{\ESC}, p) \, \mid \, q \in Q_f,\, p \in E(q) \text{ and } p^f \in E_f(q)\,\}
   \end{align*}
      \item We define the new weight function $w' : E' \to [-W',W']$ with
   $W' = -n^2W-n$ where $n=|Q|$ as follows: for every $q \not \in Q_f,  w'(e) = w(e)$, for all $e \in qE$, and for
   every $q \in Q_f$,
       \begin{align*}
         &w'(q,q_l) = w'(q,q_r) =w'(q_l,q_0) = 0,\\
         &w'(q_l, q_{\text{-nW-1}}) = - nW -1,\\
         &w'(q_r, q_{\ESC}) = n^2W + n,\\
         &w'(q_\text{-nW-1}, q') = w(q, q') \text{ for } q' \in E_f(q),\\
         &w'(q_{0}, q') = w(q, q') \text{ for } q' \in E(q),\\
         &w'(q_{\ESC}, q') = w(q, q') \text{ for } q' \in E(q)\\
       \end{align*}
   \end{itemize}

   For each $q \in Q_f$, we denote the leftmost
branch $q \to q_l \xrightarrow{-nW-1} q_\text{-nW-1}$ of the gadget as
\emph{fair branch} $\fair{q}$, the middle branch $q \to q_{l}
\xrightarrow{0} q_{0}$ as the \emph{simulation branch} $\simul{q}$ and
the rightmost branch $q \to q_r \xrightarrow{n^2W+n} q_{\ESC} $ as the
\emph{escape branch} $\escape{q}$. The intuition behind the gadgets is idential to the 1-fair case.

We now prove the correctness of
\cref{thm:mpgadgetcorrectness} for 2-fair mean-payoff games.
\begin{proof}[\Cref{thm:mpgadgetcorrectness} for 2-fair mean-payoff games]
     We first show $\Win_1(G') \, \cap \,Q \, \subseteq \,\Win_1(G)$. 
     Let $\sigma'$ be a positional player~1 strategy that wins $q$ in $G'$ for the regular mean-payoff objective. 

    Consider the subgame arena $G'_{\sigma'}$ of $G'$.  All cycles in $G'_{\sigma'}$ are non-negative; else player~2 can force to eventually only visit a negative cycle and therefore construct a play in $G'_{\sigma'}$ with negative mean-payoff value. 
    
    We construct a positional player~1 strategy $\sigma$ that wins $q$ in $G$ in a straight-forward manner: for all $q \in Q_1$, we set $\sigma(q) = \sigma'(q)$.

    For a play $\tau = q^0 q^1 \ldots \in \plays_\sigma(G)$, we define its \emph{extension} $\tau'$ in $G'$ inductively as follows: we start from with $q^0$ and $i = 0$, and with increasing $i$,
    \begin{itemize}
        \item If $q^i \not \in Q_f$, append $\tau'$ with $q^{i+1}$; Otherwise,
        \item If $\simul{q^i}$ exists in $G'_{\sigma'}$, append $\tau'$ with $q^i_l \cdot q^i_{0} \cdot q^{i+1}$;
        \item If $\simul{q^i}$ does not exist in $G'_{\sigma'}$:
        \begin{itemize} 
            \item If $q^{i+1} = \sigma'(q^i_\text{-nW-1})$, append $\tau'$ with the fair branch i.e. $q^i_l \cdot q^i_\text{-nW-1} \cdot q^{i+1}$, 
            \item Otherwise, append it with the escape branch, i.e. $q^i_r \cdot q^i_{\ESC} \cdot q^{i+1}$.
        \end{itemize}
    \end{itemize}
   % We define $\sigma$, $\sigma'$ and $G'_{\sigma'}$ as before. 
    Clearly, $\tau'$ is a play conforming wtih $\sigma'$; hence contains only non-negative cycles. Further, the restriction of $\tau'$ to the nodes $Q$ of $G$, denoted by $\tau'|_Q$, is exactly $\tau$. 
    We define the tail of $\tau$ as $\tail$ and its extension $\tail'$ similar to the proof for 1-fair mean-payoff. We denote the subgraph of $G'_{\sigma'}$ restricted to the nodes and edges in $\tail'$ by $G'_{\tail'}$.

    Let $\tau \in \plays_\sigma(G)$ be a play in $G$ that conforms with $\sigma$. If $\tau$ is not fair, it is automatically won by player~1. Therefore, assume that $\tau$ is fair. We show that all the simple cycles in $\tail$ has non-negative weight. For this, we use the fact that its extension $\tau'$ is a play in $G'_{\sigma'}$ and hence all the cycles in $G'_{\tail'}$ have non-negative weight. Since the length of simple cycles are bounded by $n$, this gives us $\meanVal(\tau) \geq 0$.

    The main argument is the absence of fair branches in $G'_{\tail'}$, introduced in~\Cref{claim:2-fair-absence}.
    
    \smallskip 
    \begin{claim}\label{claim:2-fair-absence}
      If $\tau$ is fair, there exist no fair branches $(\fair{})$ in $G'_{\tail'}$.
    \end{claim}

    \noindent Before proving the claim, we discuss how it implies that all simple cycles in $\tail$ are non-negative, resulting $\meanVal(\tau) \geq 0$. From the definition of extension, absence of fair branches in $G'_{\tail'}$ implies that either (i) for all $q \in Q_f$ only $\simul{q}$ exists in $G'_{\tail'}$, or (ii) for some $q \in Q_f$, only $\escape{q}$ exists in $\tail'$. By construction of $\tau'$, case (ii) above implies the existence of a fair node $q \in Q_f$ that appears infinitely often in $\tau$ but does not take all its fair outgoing edges, in particular $q \to \sigma'(q_\text{-nW-1}) \in E_f(q)$ infinitely often. This makes $\tau$ unfair, contradicting our assumption. In case of (i), since all simple cycles $\mathbf{c}'$ in $\tail'$ are non-negative and none of the gadget nodes contribute any weight due to only $\simul{}$ existing in $\tail'$, we conclude that the cycles $\mathbf{c} = \mathbf{c'}|_{Q}$ in $\tail$ have non-negative weight. This implies that $\meanVal(\tau) \geq 0$. 
    
    By the above mentioned argument, proving~\Cref{claim:2-fair-absence} concludes the first direction of the proof. We now prove the claim.

    \smallskip
    \noindent \textit{Proof of~\Cref{claim:2-fair-absence}.} As discussed above, whenever $\escape{q}$ is in $G'_{\tail'}$ for a node $q$, $\fair{q}$ is also in it.
    Now we remove all the escape branches from $G'_{\tail'}$ and obtain the subgraph $S$. Note that $S$ has no dead-ends, since for all $q$ for which $\escape{q}$ got removed from $G'_{\tail'}$, $\fair{q}$ was also in $G'_{\tail'}$.
    
    For all nodes $q$ in $S$, either $\fair{q}$ or $\simul{q}$ is in $S$. As a subgraph of $G'_{\sigma'}$, all cycles in $S$ have non-negative weight.
     It is easy to see that a fair branch cannot lie on a simple cycle in $S$, because fair branches carry weight $-nW-1$ and the other edges come either from $\simul{}$ carrying weight $0$, or from $G$ carrying weight $\leq W$. Hence, a simple cycle containing a fair branch would have negative weight. 
    
    Now assume fair branches exist in $\tail'$, but not on cycles of $S$. Then all cycles in $S$ consist only of simulation branches. However, nodes with simulation branches have the same outgoing edges in 
    $G'_{\tail'}$ and in $S$. Since every node in $S$ is visited infinitely often in $\tau$, $\tau$ will eventually enter cycles of $S$ and never leave them, since these nodes have the same outgoing edges in $G'_{\tail'}$. This implies fair branches are not infinitely often visited in $\tau$, and concludes the proof of~\Cref{claim:2-fair-absence}.
    \qed
    This concludes the proof of the first direction. Now let us prove the other direction, i.e. $\Win_2(G') \, \cap \,Q \, \subseteq \,\Win_2(G)$. 
     
    Once more, let $\pi'$ be a positional player~2 strategy that wins $q$ in $G'$. We construct a player~2 strategy $\pi$ that wins $q$ in $G$.
    Consider the subgame arena $G'_{\pi'}$ of $G'$.  All cycles in $G'_{\pi'}$ are negative; else player~1 can force to eventually only visit a non-negative cycle and therefore construct a play in $G'_{\pi'}$ with non-negative mean-payoff value.

    Using $\pi'$, we construct $\pi$ in $G$ as follows:
    \begin{enumerate}[label=\alph*.]
        \item \label{item:2-fair-zeroth} If $q' \not \in Q_f$, then set $\pi$ to be positional at $q'$ with $\pi(q') = \pi'(q')$. Note that this is well-defined since we do not use gadgets for nodes in $Q \setminus Q_f$.

        \item \label{item:2-fair-first} If $\pi'(q') = q_r$ for some fair node $q' \in Q_f$, then set $\pi$ to be positional at $q'$ with $\pi(q') = \pi'(q'_{\ESC})$. Intuitively this says if for a fair node $q'$, $\pi'$ asks to choose the escape branch of the gadget, then $\pi$ follows $\pi'$ to choose that branch and moves to the corresponding `escape node' in $G$. 
        
        \item \label{item:2-fair-second} If $\pi'(q') = q_l$ for some fair node $q' \in Q_f$, then we set $\pi$ to be a finite memory strategy that is played in rounds, using a local memory of size $O(n^3W + n^2 + n + |E_f(q')| + 1)$ at $q'$: consider an order on the fair successors of $q'$ and keep a pointer $j$ in the set $E_f(q')$. In round $i \geq 0$, $\pi$ choses the `preferred successor' $\pi'(q_0')$ if $i \mod n^3W + n^2 + n + 1 \neq 0$. Otherwise, it takes the $j^{\text{th}}$ (modulo $|E_f(q')|$) fair edge of $q'$ and increments $j$.
    \end{enumerate}

    \noindent Now we prove that all $\tau \in \plays_\pi(G)$ are winning for player~2 in $G$. In particular, we will show that $\tau$ is fair and $\meanVal(\tau) < 0$.
    
    \smallskip
    For a play $\tau = q^0 q^1 \ldots \in \plays_\pi(G)$, we construct its \emph{extension} $\tau'$ in $G'$ inductively as follows: we start with $q^0$ and $i=0$ and with increasing $i$,
    \begin{itemize}
      \item If $q^i \not \in Q_f$, append $\tau'$ with $q^{i+1}$;
      \item If $\pi(H\cdot q^i)$ is defined using~\Cref{item:2-fair-first} i.e., $\pi'(q^i) = q^i_r$ then append $\tau'$ with $\escape{q^i} \to q^{i+1}$. In particular, we append $\tau'$ with $q^i_r \cdot q^i_{\ESC}\cdot q^{i+1}$.
      \item If $\pi(H \cdot q^i)$ is defined using~\Cref{item:2-fair-second}, then
        \begin{itemize}
          \item if $q^{i+1} = \pi'(q^i_{0})$, append $\tau'$ with $\simul{q^i}\to q^{i+1}$ i.e., with $q^i_l \cdot q^i_{0} \cdot q^{i+1}$ ; 
          \item else, append $\tau'$ with $\fair{q^{i}} \to q^{i+1}$, i.e. $q^i_l \cdot q^i_\text{-nW-1} \cdot q^{i+1}$.
        \end{itemize} 
    \end{itemize}

    Observe that again, $\tau'$ is a play in $G'_{\pi'}$; hence it contains only negative cycles, and $\tau'|_Q = \tau$. 
    We define a tail $\tail'$ of $\tau'$ as well as $\tail$ and $G'_{\tail'}$, as before.

    We first show that $\tau$ is fair. It is sufficient to show that $\tail$ does not contain fair nodes $p$ for which $\pi(p)$ is defined via~\Cref{item:2-fair-first}, because for all the other fair nodes in $\tail$, $\tau$ visits all their fair successors infinitely often (\Cref{item:2-fair-second}).
     Let $p$ be a node for which $\pi(p)$ is defined via~\Cref{item:2-fair-first}. Then $\pi'(p) = p_r$. As $\pi'$ is positional, only the escape branch of $p$ exists in $G'_{\tail'}$, carrying the weight $n^2W + n$. 
    Since $p$ is visited infinitely often, there is a (simple) cycle in $G'_{\pi'}$ that passes through $p$, and thus sees the weight $n^2W +n$. In order to maintain a non-negative weight, the cycle needs to contain $n$-many fair branches (carrying weight $-nW-1$), which is not possible since there are at most $n-1$ many fair branches in $G'_{\pi'}$, due to $p$ taking its escape branch. With this we conclude that $\tau$ is fair.
    
    Now we prove that $\meanVal(\tau) < 0$. First let us make an observation about the cycles in $G'_{\pi'}$ that use only simulation branches. 

    \smallskip
    \begin{observation}[sim-cycles] We call cycles $\mathbf{c'}$ in $G'_{\pi'}$ that do not contain any fair or escape branches \emph{sim-cycles}. That is, $\mathbf{c'}$ contains only simulation branches. Being part of $G'_{\pi'}$, $\mathbf{c}$ has negative weight.
    Since all the gadget branches in $\mathbf{c'}$ carry $0$ weight, the restriction $\mathbf{c} = \mathbf{c'}|_Q$ is a negative weight cycle in $G$.
    \end{observation}
    
    We saw that $\tail$ does not contain nodes defined using~\Cref{item:2-fair-first}. I.e., for all fair nodes $p$ in $\tail$, $\pi$ is defined via~\Cref{item:2-fair-second}, and $G'_{\tail'}$ contains the branches $\fair{p}$ and/or $\simul{p}$. 
    %IRMAK LEFT HERE
    
   If $\tail$ does not contain any fair nodes, $\tail = \tail'$, and $\meanVal(\tau) < 0$ easily follows. Now we assume $\tail$ contains at least one fair node. Then in order to show $\meanVal(\tau)< 0$, it is sufficient to show the existence of a number $k\in \mathbb{N}$ and a positive rational number $\epsilon$ such that every $k$-length infix of $\tail$ has average weight smaller than $-\epsilon$.
    We show this for $k = n^3W + n^2 + n$ and $\epsilon = 1/k$.

    Observe that every $k$-length infix $\uu'$ of $\tail'$ contains at least $(n^2W + n +1)$ simple cycles, and the length of $\uu = \uu'|_Q$ is at most $k$. Moreover, every fair branch is visited at most once in $\uu'$ (due to~\Cref{item:2-fair-second}). Since $\tail'$ does not contain escape branches, every simple cycle in $\tail'$ is either a sim-cycle, or contains a fair branch. Since each fair branch is seen at most once in $\uu'$, at most $n$ of the cycles $\tail'$ contains are not sim-cyles. So, $\tail'$ contains at most $n$ non-sim-cycles, and at least $n^2W + 1$ sim-cycles. Since the projection $\mathbf{c} = \mathbf{c'}|_Q$ of a sim-cycle $\mathbf{c'}$ to $\uu$ has weight at most $-1$; and the projection of a non-sim-cycle has weight at most $nW$; $\uu$ has weight at most $-1$. This concludes that every $k$-length infix of $\tail$ has averge weight at most $-1/k$. \qed

    This concludes the proof of the second direction and that of the theorem. 

\end{proof}

\section{Proof of \Cref{thm:2-fairEnergyWin1}}
% {Proofs of \Cref{sec:2-fair-energy}}
\label{sec:app-proofs-2-fair-energy}
% \subsection{Proof of \Cref{thm:2-fairEnergyWin1}}

\begin{proof}
    Let $q$ be a node in $G$. It is easy to see that, if player~1 wins from $q$ in the regular energy game, then she also wins from $q$ in the 2-fair energy game following the same strategy. 
  
  For the opposite side, we show that, if player~1 does not win from $q$ in the regular energy game, then she does not win from $q$ in the 2-fair energy game. We prove it by contradiction. Assume that player~1 does not win from $q$ in $(G, \energy)$ but wins from $q$ in $(G, E_f, \energy)$. Then, there exists a finite initial credit $c$ and player~1 strategy $\sigma \in \Sigma_1$ such that for all player~2 strategies $\pi \in \Sigma_2$, the play $\play_{\sigma,\pi}(G, q)$ either satisfies $\energy_c$ or it is not fair. On the other hand, positional determinacy of energy games implies that there exists a positional player~2 strategy $\pi'$ in $(G, \energy)$ that wins from $q$, i.e., all the plays conforming to $\pi'$ violates energy. Consider the player~2 strategy $\pi'_c$ in $(G,E_f,\energy)$ that mimics $\pi'$ in $(G, E_f, \energy)$ until the accumulated energy of a play drops below $0$ and then play fairly (satisfy the fairness condition). Note that, this is always possible since for a fixed initial credit $c$, the energy of a play drops below $0$ after finitely many steps (due to determinacy of energy games), and fairness is a prefix-independent objective. Clearly, the play $\play_{\sigma,\pi'_c}(G, q)$ is fair and it does not satisfy $\energy_c$, which leads to a contradiction. Also note that, this argument does not make the node $q$ winning for player~2 in the 2-fair energy game, since the strategy $\pi'_c$ depends on the finite initial credit $c$.\qed
  \end{proof}

\section{Proofs of~\Cref{sec:1-fair-energy}}\label{sec:app-proofs-P1-fair-energy}

\subsection{Proof of~\Cref{thm:determinacy1-fairenergy}}

    For this proof, we will need to define \emph{strategy trees}.

    \smallskip
    \noindent \textbf{Strategy Trees.} The \emph{strategy tree} $\tree_{\pi(q)}$ of a player~$i$ strategy $\pi$ is obtained by unfolding the game $(G, E_f)$ starting from node $q$, according to $\pi$. The nodes of the tree are prefixes $u_1 \ldots u_k$ of $\plays_\pi(G, q)$. Therefore, each infinite branch corresponds to a play in $\plays_\pi(G, q)$. We will use nodes (branches) and the finite prefixes (plays) they represent, interchangably. Furthermore, we say the node $u_1 \ldots u_k$ has weight $w(u_1 \ldots u_k)$.

    In a 1-fair energy game a player~1 strategy $\sigma$ is winning from $q$ iff there exists a $c \in \mathbb{N}$ such that all branches of $\tree_{\sigma(q)}$ are fair plays whose weights don't drop below $c$.
    A player~2 strategy $\pi$ is winning from $q$ if all branches of $\tree_{\pi(q)}$ are plays that are either unfair, or whose weights drop unboundedly (i.e. for each $k \in \mathbb{N}$, the weight of each fair branch in $\tree_{\pi(q)}$ drops below $-k$ at some point). 
    Observe that strategies and strategy trees for player~$i$ on game $G$ have a 1-to-1 correspondance.

  \begin{proof}[\Cref{thm:determinacy1-fairenergy}]
    % We show that as in regular energy games, we can restrict the quantification over $c \in \mathbb{N}$ to a finite set (namely to $c \in [0, nW]$) in \Cref{eq:energyDeterminacyPlayer2-1-equivalent}. That is, we show that if there exists a player~2 strategy $\sigma$ that satisfies $\forall \pi \in \Sigma_1, \play_{\sigma, \pi}(G, q) \not \in \energy^f_{nW}$, then there exists a winning player~2 strategy $\sigma^{\ext}$. As in the discussion in~\Cref{subsec:determinacyoffairenergy}, this will give us the equivalence of~\Cref{eq:energyDeterminacyPlayer2-1} and~\Cref{eq:energyDeterminacyPlayer2-2}; thus ensuring determinacy.
%
     Let $q \in Q \setminus \Win_1$. Then, \Cref{eq:energyDeterminacyPlayer1,eq:energyDeterminacyPlayer2} holds. That is, for each $c \in \mathbb{N}$, there exists a player~2 strategy $\pi_c$ that wins against all player~1 strategies. In particular for $c = nW$, there exists such a $\pi_{nW}$. Then the weight of all the fair branches of $\tree_{\pi_{nW}(q)}$ drop below $-nW$ at some point. We will extend $\tree_{\pi_{nW}(q)}$ to obtain a winning player~2 strategy tree $\tree_{\pi^\ext(q)}$. 
   
    Let $\rho$ be a fair branch in $\tree_{\pi_{nW}(q)}$. Then we are guaranteed to 
    find a prefix $\pr \cdot \w$ of $\rho$ such that $\w$ is a simple loop on a state $p$, $w(\pr)$ is non-positive and $w(\w)$ is negative. To see this, observe that $\rho$ has a prefix $\uu = u_1 \ldots u_k$ with $w(\uu) < -nW$. For each $i \leq k$, we denote $u_1 \ldots u_i$ with $\uu_i$. 
    Let $\uu_i$ be the longest prefix of $\uu$ with positive weight. Then, $w(\uu_{i+1}) \leq 0$. 
     Since $w(u_{i} \ldots u_k) < -nW$, it contains a negative weight simple cycle. 
   Let $u_j \ldots u_{j'}$ be the first such cycle with $u_j = u_{j'} = p$. Since $j \geq i$, clearly $w(\uu_j)$ is non-positive. Then, we set $\pr = \uu_j$, and $\w = u_j \ldots u_{j'}$.
   Therefore, every fair branch $\rho$ of $\tree_{\pi_{nW}(q)}$ has a unique signature, $(\pr_\rho, \w_\rho)$.
   
   Now we are ready to build $\tree_{\pi^\ext(q)}$. We build the tree roughly by cutting each fair branch at $\pr\cdot \w$ and pasting the $\w$ part from there on.
   
   Formally, we start building $\tree_{\pi^\ext(q)}$ through $\tree_{\pi_{nW}(q)}$, branch by branch. For a finite prefix $\uu$ of each branch,
   $\pi^\ext$ prescribes the same successors to player~2 nodes as $\tree_{\pi_{nW}(q)}$, until
   the finite prefix $\uu$ is equal to $\pr \cdot \w$ for a fair branch that has $\uu$ as a prefix. That is, the trees are identical until some $\pr \cdot \w$ is encountered. Then, for every player~2 node $u \in \w$, $\pi^\ext$ prescribes the unique successor of $u$ in $\w$ to all $u$ with history $\pr \cdot (\w)^i \cdot \w'$ where $i \geq 1$ and $\w'$ a prefix of $\w$. 
   So, $\pi^\ext$ mimics the strategy of $\w$ for all plays with history $\pr \cdot (\w)^i$. 
   For player~2 nodes with history $\pr \cdot (\w)^i \cdot \w'$ where $\w'$ is not a prefix of $\w$; 
   $\pi^\ext$ identifies these histories with $\pr \cdot \w'$ and prescribe such player~2 nodes, the same successors as nodes with history $\pr \cdot \w'$. Note that, if $\pr \cdot \w'$ is a prefix of a fair branch $\rho$ with signature $(\pr_\rho, \w_\rho)$, then $\pr$ is a prefix of $\pr_\rho \cdot \w_\rho$. If this was not the case, while assigning the strategy to $\pr$ in $\tree_{\pi^\ext(q)}$, we would mimic the strategy of $\rho$ wrt $(\pr_\rho, \w_\rho)$. 
   
   Let us look at branches $\rho'$ in $\tree_{\pi^\ext(q)}$. Then $\rho'$ satisfies one of the following three cases:
   \begin{enumerate}[label=(\roman*)]
     \item It eventually mimics a branch $\rho$ in $\tree_{\pi_{nW}(q)}$,\label{item:p1}
   \item It has a tail of $(\w)^\omega$ where $(\pr, \w)$ is the signature of a fair branch $\rho$ in $\tree_{\pi_{nW}(q)}$, \label{item:p2}
   \item It visits cycles $\w_1, \w_2, \ldots $ that belong to signatures  $(\pr_1, \w_1),  (\pr_2, \w_2),  \ldots$ of fair branches $\rho_1, \rho_2, \ldots $ infinitely often, but it does not have a tail $(\w_i)^\omega$ that cycles in only one of them. That is, $\rho'$ is of the form $\w'_1 \cdot (\w_1)^{j_1} \cdot \w'_2 \cdot (\w_2)^{j_2} \cdot \w'_3 \cdot (\w_3)^{j_3} \cdots$ for $j_k \neq 0$ for all $k \in \mathbb{N}$ and all $\w'_i$ are minimal, in the sense that $\w'_i$ does not contain $\w_{i-1}$ as a prefix or $\w_i$ as a suffix.\label{item:p3}
   \end{enumerate}
   In case of \Cref{item:p1} $\rho'$ is either unfair, or it eventually mimics a play of the form $\pr \cdot (\w)^\omega$ in $\tree_{\pi_{nW}(q)}$. Since $w(\w)$ is negative, the weight of $\rho'$ drops unboundedly.
   Similarly, in case of \Cref{item:p2}, the weight of $\rho'$ drops unboundedly.
   
   For the rest of the proof, we will show that the weight of $\rho'$ drops unboundedly in case of \Cref{item:p3}. Recall that for all $i$, $\pr_i$ is a prefix of $\pr_{i+1} \cdot \w_{i+1}$. 
   %Recall that the history of $\pr_i \cdot (\w_i)^{j_i} \cdot \w'_i$ is identified with the history  $\pr_i  \w'_i$. Then inductively, the history $\w'_1 \cdot (\w_1)^{j_1} \cdot \w'_2 \cdot (\w_2)^{j_2} \cdots \w'_k \cdot (\w_k)^{j_k} \cdot \w'_{k+1}$ is identified with the history $\w'_1 \cdot \w'_2 \cdots \w'_k \cdot \w'_{k+1}$. The fact that 
   Furthermore, the history $\pr_i \cdot (\w_i)^{j_i} \cdot \w'_{i+1}$ is identified with the history $\pr_i \cdot \w'_{i+1}$ in $\pi^\ext$. The fact that the history $\pr_i \cdot (\w_i)^{j_i} \cdot \w'_{i+1}$ is continued by $\w_{i+1}$ in $\tree_{\pi^\ext}$ implies that $\pr_i \cdot \w'_{i+1}$ is equal to $\pr_{i+1}\cdot (\w_i)^k$ for some $k \in \mathbb{N}$.
   
   In order to finish the proof, we will introduce the concept of shuffle-equivalence and prove \Cref{claim:cyclefreerho}.  
   Two finite plays $\uu_1$ and $\uu_2$ are called \emph{shuffle-equivalent} iff they see the same cycles, the same number of times. Clearly, the weights of two shuffle-equivalent finite plays are equivalent.
   
   We denote $\rho'$ without the cycles $\w_i$ as $\rho^- = \w'_1 \cdot \w'_2 \cdots$ and denote the $k$-element prefix $\w'_1 \cdot \w'_2 \cdots \w'_k$ of $\rho^-$ with $\rho^{-}_k$. 
   We show that $\rho^-_k$ has non-positive weight for every $k$, which directly follows from the following claim.
   \begin{claim} \label{claim:cyclefreerho}
     $\rho^-_k$ is shuffle equivalent to $\pr_{k+1} \cdot (\w_1)^{z_1} \cdot (\w_2)^{z_2} \cdots (\w_{k+1})^{z_{k+1}}$ 
     for some $z_1, \ldots, z_k \in \mathbb{N}$.
   \end{claim}
   \begin{proof}[\Cref{claim:cyclefreerho}] We will proceed by induction on $k$. The base case for $k = 1$ follows from our previous observation that $\pr_1 \cdot \w'_2$ is equivalent to $\pr_{2} \cdot (\w_2)^{z_2}$ for some $z_2 \in \mathbb{N}$. 
   We assume the claim holds for $k$, and show it holds for $k+1$.
   By the inductive hypothesis, $\rho^-_{k+1}$ is shuffle equivalent to $\pr_{k+1} \cdot (\w_1)^{z_1} \cdot (\w_2)^{z_2} \cdots (\w_{k+1})^{z_k} \cdot \w'_{k+1}$. Then the claim follows the observation that $\pr_{k+1} \cdot \w'_{k+1}$ is of the form $\pr_{k+2} \cdot (\w_{k+2})^{z_{k+2}}$ for some $z_{k+2} \in \mathbb{N}$. 
   \end{proof}
   
   \Cref{claim:cyclefreerho} immediately implies that $\rho^-_k$ has non-negative weight for each $k$, since $\pr_i$ and $\w_i$ both have non-positive weight. This in turn implies that $w(\rho'_k) \leq -k$ where $\rho'_k = \w'_1 \cdot (\w_1)^{j_1}\cdot \w'_2 \cdots \w'_k \cdot (\w_k)^{j_k}$ since all $j_i$ are positive. This shows that the weight of $\rho'$ drops unboundedly, which finishes the proof by revealing that $\pi^{\ext}$ is a winning player~2 strategy from $q$. \qed
   \end{proof}

\subsection{Proof of~\Cref{thm:energygadgetcorrectness}}

   Let $(G, E_f, \energy^f)$ be a 1-fair energy game. We
   construct $G'$ by replacing every fair node in $q \in Q_f$ in $G$, with
   its corresponding gadget presented in \cref{fig:1-fair-energy-gadget}.
   Formally, the `gadget game' $G'= (Q',E',w')$ with $Q' = Q'_1 \uplus
   Q'_2$ is defined as follows: 

\begin{itemize}
    \item For every fair node $q \in Q_f$, we introduce three new player~1
  nodes $q_{\texttt{val}}, q_{\texttt{pos}},  q_{\ESC}$ and four new player~2 nodes $q_l, q_r,
  q_{\text{nW+1}}, q_{\texttt{0}}$ in $G'$.
      \begin{align*}
          Q'_1 = &Q_1 \cup \{q_{\texttt{val}}, q_{\texttt{pos}},  q_{\ESC} \, \mid \, q \in Q_f\}\\
          Q'_2 = &Q_2 \cup \{q_l, q_r, q_\text{nW+1}, q_{\texttt{0}}\, \mid \, q \in Q_f\}
      \end{align*}
     \item With the introduction of the new nodes in the gadgets, we
  introduce new edges as follows:
     \begin{align*}
      E' \,= &\,\, \{\,qE  \,\mid\, q \not \in Q_f\,\}\\
      &\,\cup \,\{\, (q, q_l),\, (q, q_r), \,(q_l, q_\text{nW+1}),\, (q_l,
  q_{\texttt{val}}),\,(q_{\texttt{val}}, q_{\texttt{pos}}), \, (q_{\texttt{val}}, q_{\texttt{0}})\, (q_r, q_{\ESC}) \,\mid \, q \in Q_f\,\} \\
      &\,\cup \,\{\, (q_\text{nW + 1}, q'), \,  (q_\texttt{0}, q') \mid q' \in E_f(q)\} \cup
  \{(q_{\texttt{pos}}, q'),\, (q_{\ESC}, q')\mid q' \in E(q)\}
     \end{align*}
     \item We define the new weight function $w' : E' \to [-W',W'] \cup \{-\varepsilon\}$ with $\varepsilon = \frac{1}{n+1}$ and
  $W' = n^2W+n$ where $n=|Q|$ as follows: for every $q \not \in Q_f,  w'(e) = w(e)$, for all $e \in qE$, and for
  every $q \in Q_f$,
      \begin{align*}
        &w'(q,q_l) = w'(q,q_r) =w'(q_l,q_{\texttt{val}}) = w'(q_\texttt{val}, q_{\texttt{0}}) = 0,\\
        &w'(q_l, q_{\text{nW+1}}) = nW + 1,\\
        &w'(q_{\texttt{val}}, q_{\texttt{pos}}) = - \epsilon,\\
        &w'(q_r, q_{\ESC}) = -n^2W - n,\\
        &w'(q_\text{nW+1}, q') = w'(q_{\texttt{0}})= w(q, q') \text{ for } q' \in E_f(q),\\
        &w'(q_{\texttt{pos}}, q') = w'(q_{\ESC}, q') = w(q, q') \text{ for } q' \in E(q)\\
      \end{align*}
  \end{itemize}

  For each $q \in Q_f$, we denote the leftmost
branch $q \to q_l \xrightarrow{nW+1} q_\text{nW+1}$ of the gadget as
\emph{fair branch} $\fair{q}$, the middle branch $q \to q_{\texttt{val}}$ as the \emph{value branch} $\valu{q}$ and
the rightmost branch $q \to q_r \xrightarrow{-n^2W-n} q_{\ESC} $ as the
\emph{escape branch} $\escape{q}$. Furthermore, we denote the left branch $q \to q_\texttt{val} \xrightarrow{-\epsilon} q_{\texttt{pos}} $ of the value branch 
as the \emph{positive-value branch} $\pos{q}$ and its right branch $q \to q_\texttt{val} \xrightarrow{0} q_\texttt{0} $ as the \emph{zero-value branch} $\zero{q}$. 

We note that the game graph as formally defined above, is equivalent to a game graph on integer weights obtained by multiplying every weight by $\epsilon = \frac{1}{n+1}$. The size of the gadget game $G'$ is stated in~\Cref{thm:energygadgetcorrectness} is calculated with respect to the integer weights. For simplicity, in the formal definition and the proof we use the rational weights depicted on the figure.

We now formally prove the correctness of
\cref{thm:energygadgetcorrectness}.  
\begin{proof}

    We first show that $\Win_1(G') \, \cap \,Q \, \subseteq \,\Win_1(G)$ i.e., if player~1 wins the regular energy game from a node $q$ in $G'$, then player~1 also wins the fair energy game from $q$ in $G$ if $q$ is in $G$. 

    Let $\sigma'$ be a winning strategy for player~1 in $G'$ for the regular energy objective starting from node $q$ with $q \in \Win_1(G') \cap Q$. Without loss of generality, let $\sigma'$ be a positional strategy. We will construct a winning strategy $\sigma$ for player~1 in $G$ for the fair energy objective starting from node $q$.

    Consider the subgame arena $G'_{\sigma'}$ of $G'$ with respect to the strategy $\sigma'$. Similar to the mean-payoff case, all cycles in $G'_{\sigma'}$ are non-negative; else player~2 can force to eventually only visit a negative cycle and therefore construct a play in $G'_{\sigma'}$ whose energy drops unboundedly. This leads to a contradiction since $\sigma'$ is winning for player~1.
    
    Using $\sigma'$, we construct $\sigma$ in $G$ as follows:
    \begin{enumerate}[label=\alph*.]
        \item \label{item:1-fair-energy-zeroth} If $q' \not \in Q_f$, then set $\sigma$ to be positional at $q'$ with $\sigma(q') = \sigma'(q')$. Note that this is well-defined since we do not use gadgets for nodes in $Q \setminus Q_f$.

        \item \label{item:1-fair-energy-first} If $\sigma'(q') = q_r$ for some fair node $q' \in Q_f$, then set $\sigma$ to be positional at $q'$ with $\sigma(q') = \sigma'(q'_{\ESC})$. Intuitively this says if for a fair node $q'$, $\sigma'$ asks to choose the escape branch of the gadget, then $\sigma$ follows $\sigma'$ to choose that branch and moves to the corresponding `escape node' in $G$. 
        
        \item \label{item:1-fair-energy-second} If $\sigma'(q') = q_l$ and $\sigma(q_\texttt{val}) = q_\texttt{pos}$ for some fair node $q' \in Q_f$, then we set $\sigma$ to be a finite memory strategy that is played in rounds, using a local memory of size $\log(n^3W  + |E_f(q)| + 1)$ at $q'$: consider an order on the fair successors of $q'$ and keep a pointer $j$ in the set $E_f(q')$. In round $i \geq 0$, $\sigma$ chooses the `preferred sucessor' $\sigma'(q'_0)$ if $i \mod n^3W + 1\neq 0$. Otherwise, it takes the $j^{th}$ (modulo $|E_f(q')|$) fair edge of $q'$ and increments $j$.  
        
        \item \label{item:1-fair-energy-third} If $\sigma'(q') = q_l$ and $\sigma(q_\texttt{val}) = q_\texttt{0}$ for some fair node $q' \in Q_f$, then we let $\sigma$ fix an order on the fair successors of $q'$ as before, and loop through this set. Intuitively this choice of $\sigma$ implies that all fair successors of $q'$ lead to non-negative value cycles. 
    \end{enumerate}

    \noindent Recall that in 1-fair energy games, player~2 wins a node $q$ iff it has a winning strategy $\pi$ that wins against all initial credits $c$ and all player~1 stragies (\Cref{eq:energyDeterminacyPlayer2}). Therefore, the weight in all plays in $\plays_\pi(G)$ drops unboundedly. 
    This in turn implies that a player~1 strategy $\lambda$ is winning if all plays in $\plays_\lambda(G)$ are fair and have energies bounded from below (not necessarily by the same initial credit $c$).

    Now we prove that all $\sigma$ plays $\tau \in \plays_\sigma(G)$ are winning for player~1 in $G$. In particular, we will show that (i) $\tau$ is fair and (ii) $ \tau = \mathbf{w} \cdot \tail$ for a finite prefix $\mathbf{w}$ and an infinite suffix $\tail$ whose weight do not drop below $-n^2W$. Observe that (ii) implies that the weight of $\tau$ is bounded from below. 
   
   \smallskip
   For a play $\tau = q^0 q^1 \ldots \in \plays_\sigma(G)$, we construct its \emph{extension} $\tau'$ in $G'$ inductively as follows: we start with $q^0$ and $i=0$ and with increasing $i$,
   \begin{itemize}
     \item If $q^i \not \in Q_f$, append $\tau'$ with $q^{i+1}$;
     \item If $\sigma(H\cdot q^i)$ is defined using~\Cref{item:1-fair-energy-first} i.e., $\sigma'(q^i) = q^i_r$ then append $\tau'$ with $\escape{q^i} \to q^{i+1}$. In particular, we append $\tau'$ with $q^i_r \cdot q^i_{\ESC}\cdot q^{i+1}$.
     \item If $\sigma(H \cdot q^i)$ is defined using~\Cref{item:1-fair-energy-second},  $\sigma'(q^i) = q^i_l$ and $\sigma'(q^i_{\texttt{val}}) = q^i_{\texttt{pos}}$ then
       \begin{itemize}
         \item if $q^{i+1} = \sigma'(q^i_{\texttt{pos}})$, append $\tau'$ with $\pos{q^i}\to q^{i+1}$ i.e., with $q^i_l \cdot q^i_{\texttt{val}}  \cdot q^i_{\texttt{pos}} \cdot q^{i+1}$; 
         \item else, append $\tau'$ with $\fair{q^{i}} \to q^{i+1}$, i.e. $q^i_l \cdot q^i_\text{nW+1} \cdot q^{i+1}$.
       \end{itemize} 
    \item If $\sigma(H \cdot q^i)$ is defined using~\Cref{item:1-fair-energy-third}, $\sigma'(q^i) = q^i_l$ and  $\sigma'(q^i_{\texttt{val}}) = q^i_{\texttt{0}}$ then append $\tau'$ with $\zero{q^i} \to q^{i+1}$, i.e., with $q^i_l \cdot q^i_{\texttt{val}}  \cdot q^i_{\texttt{0}} \cdot q^{i+1}$.
   \end{itemize}

   Clearly, $\tau'$ is a play in $G'_{\sigma'}$ and hence it contains only non-negative cycles. Furthermore, the restriction of $\tau'$ to the nodes $Q$ of $G$, denoted by $\tau'|_Q$, is exactly $\tau$. 

   Let a `tail' of $\tau$ denote the infinite suffix $\tail$ of $\tau$ that contains only the nodes and edges visited infinitely often in $\tau$; and let $\tail'$ be the extension of $\tau$ in $G'$. We denote the subgraph of $G'_{\sigma'}$ restricted to the nodes and edges in $\tail'$ by $G'_{\tail'}$.

   The proof of (i), i.e. that $\tau$ is fair, is exactly the same as in 1-fair mean-payoff games. To show fairness of $\tau$, it is sufficient to show that $\tail$ does not contain fair nodes $p$ for which $\sigma(p)$ is defined as~\Cref{item:1-fair-energy-first}. Let $p$ be a node for which $\sigma(p)$ is defined via~\Cref{item:1-fair-energy-first}. Then $\sigma'(p) = p_r$. As $\sigma'$ is positional, among all the branches from the gadget of $p$, only $\escape{p}$ exists in $G'_{\sigma'}$ carrying the weight $-n^2W - n$. 
   Since from our assumption $p$ is visited infinitely often, there is a (simple) cycle in $G'_{\sigma'}$ that passes through $p$, and thus sees the weight $-n^2W -n$. This cycle must be non-negative
   since it is a cycle in $G'_{\sigma'}$, but in order to counteract the negative weight $-n^2W -n$, the cycle needs to contain $n$-many gadget edges with weight $nW+1$, which is not possible. Therefore, $\tau$ is fair.

   Now we prove that the weight of the prefixes $\tail_i$ of $\tail$ do not go below $-n^2W$. Similar to the proof of 1-fair mean-payoff, we make an observation about cycles in $G_{\tail'}$ that do not use left branches. Different from the mean-payoff case, we have 2 different kinds of simple cycles that do not see a left branch. One of them sees at least one positive value branch, and therefore collects weight $-\epsilon$. Since the cycle have non-negative value in $G'$, and the gadget nodes contribute only negative weights; we projection of this cycle in $G$ have positive weight. The second kind of cycles only see zero value branches. The projection of such cycles in $G$ can have zero or positive weight.

   \smallskip
    \begin{observation}[positive-value and zero-value cycles] Let $\mathbf{c'}$ be a cycle in $G'_{\sigma'}$ that does not contain any fair or escape branches. If $\mathbf{c'}$ sees a positive value branch, we call it a \emph{positive-value cycle}; else, we call it a \emph{zero-value cycle}. For both positive- and zero-value cycles $\mathbf{c'}$, since the cycles have non-negative weight and the gadget nodes do not contribute positively to the weight, $\mathbf{c'}|_{Q}$ has non-negative weight. Furthermore in a positive-value cycle $\mathbf{c'}$ the gadget nodes contribute negatively to the weight of the cycle. Therefore, $w(\mathbf{c'}_{Q}) \geq 1$.
    \end{observation}

    We have shown above that $\tail$ contains no nodes for which $\sigma$ has been defined via~\Cref{item:1-fair-first}. Therefore, for all fair nodes $p$ in $\tail$, $\sigma$ is defined via~\Cref{item:1-fair-energy-second} or~\Cref{item:1-fair-energy-third}, and $G'_{\tau'}$ contains only the branches $\fair{p}$ and/or $\valu{p}$. %Observe that if for a node $p$, $\sigma$ is defined via~\Cref{item:1-fair-energy-second} then $G'_{\tail'}$ contains only the branches $\fair{p}$ and $\pos{p}$. On the other hand if for $p$, $\sigma$ is defined via~\Cref{item:1-fair-energy-third} then $G'_{\tail'}$ contains only the branch $\zero{p}$.\ISinside{Do i need this?} %Let $z$ be the number of fair nodes $p$ in $\tail$ for which $\sigma$ is defined via~\Cref{item:1-fair-third}. Then we know that $G'_{\tail'}$ contains fair branches of at most $n-z$ nodes.    

   If no fair branches exist in $\tail'$, then for no $p \in \tail$, $\sigma$ is defined via~\Cref{item:1-fair-energy-third}. Therefore the only difference beween $\sigma'$ and $\sigma$ are the zero-value branches which do not alter the weight of the game. The fact that the weight of $\tail$ is bounded below easily follows. 
   
   Otherwise, there exists a node $p \in \tail$ such that $\fair{p}$ is in $G'_{\tail'}$.
   We fix such a node $p$, and show in~\Cref{claim:boundedness-of-infix} that for every infix $\uu'$ of $\tail'$ that corresponds to a simple cycle over $\fair{p}$ in $\tail'$, $\uu = \uu'|_Q$ has positive weight. Moreover, the weight of $\uu$ does not drop below $-n^2W$ for any prefix. The proof of the claim concludes that the weight of $\tail$ is bounded below by $-n^2W$.
   
   \begin{claim}\label{claim:boundedness-of-infix}
     Let $\uu'$ be an infix of $\tail'$ which starts and ends at node $p$, the first branch of $\uu'$ as well as the branch following $\uu'$ in $\tail'$ are $\fair{p}$. Then $w(\uu) > 0$, for $\uu = \uu'|_Q$; and $w(\uu_i) \geq -n^3W$ for every prefix $\uu_i$ of $\uu$. 
   \end{claim}

   \begin{proof}[\Cref{claim:boundedness-of-infix}]
   Observe that since $\uu'$ is a simple cycle on $\fair{p}$; the fair branch of $p$ is seen only once in $\uu'$, i.e. as the first branch. However, the fair branches of the other nodes can repeat arbitrarily often in $\uu'$. We will proceed by induction on the number of such nodes. 
   We will show the following: If there are $\leq (n-k)$-many nodes whose fair branches are taken more than once in $\uu'$ for some $ k \in [1, n]$, then $w(\uu) \geq (k-1) n^2 W$ and $w(\uu_i) \geq -(n-k+1)n^2W$ for every prefix $\uu_i$ of $\uu$.
    
   \textbf{Base case (k = n):} In this case, the fair branches of all nodes are seen at most once in $\uu'$. That is, there are at most $n$-many fair branches in $\uu'$, and the rest are positive- or zero-value branches. Consequently, in $\uu'$ there are at most $n$-many simple cycles that are neither positive- nor zero-value cycles. These are the only cycles whose projections to $Q$ can have negative weight in $\uu$; and since they are simple cycles they can collect the minimal weight of $-n^2W$. Therefore, the weight of any infix of $\uu$ is lower bounded by $-n^2W$. Furthermore, we know that the positive-value branch $\pos{p}$ is visited exactly $n^3W$-many times in $\uu'$. Therefore, there are at least $n^3W$-many positive-value cycles in $\uu'$, collecting the minimum weight of $n^3W$ in $\uu$. Therefore, $w(\uu) \geq n^3W - n^2W = (n-1) n^2 W$.

   \smallskip
   \textbf{Inductive Hypothesis:} Assume the statement holds for $i \in [k+1, n]$, i.e. if $\uu'$ has at most $n-i$ different fair branches repeating for some $i \in [k+1, n]$, then $w(\uu) \geq (i-1) n^2 W$ and $w(\uu_i) \geq -(n-i+1) n^2W$ for any prefix $\uu_i$ of $\uu$. 

   \smallskip 
   \textbf{Inductive Step:} We will show that the statement holds for $k < n$. 
   WLOG we assume exactly $(n-k)$-many different fair branches repeat in $\uu'$, since otherwise the desired inequality follows from the inductive hypothesis.
   We can find infices $\mathbf{w^1}, \ldots, \mathbf{w^t}$ of $\uu'$ for some $t \in [1, n-k]$ such that each $\mathbf{w^i}$ 
   is a (not necessarily simple) cycle on the fair branch $\fair{p_i}$ of a unique node $p_i$. Note that the fair nodes of $\mathbf{w^i}$ and $\mathbf{w^j}$ are different for $i \neq j$. 
   Furthermore, we can repeat this operation exhaustively such that if we remove the infices $\mathbf{w^1}, \ldots, \mathbf{w^t}$ from $\uu'$, in the remaining sequence $\uu'_{-}$%\ISinside{need to set up this (or some better) notation for this} 
   there are no repeating fair branches.    
   
   Observe that $\fair{p_i}$ can be visited multiple times in $\mathbf{w^i}$. Therefore, we view each $\mathbf{w^i}$ as the non-empty concetanation 
   $\mathbf{w^i}_1\cdot \mathbf{w^i}_2 \cdots \mathbf{w^i}_r$ where each $\mathbf{w^i}_j$ is a simple cycle on $\fair{p_i}$ for $j \in [1, r]$.
   Then the fair branches of at most $(n-(k+1))$-many nodes repeat in $\mathbf{w^i}_j$. By inductive hypothesis, 
   $w(\mathbf{w^i_j}|_{Q}) \geq k n^2 W$, and the weight of $w^i_j|_Q$ does not drop below $-(n-k)n^2W$. Consequently, $w(\mathbf{w^i}|_Q) \geq k n^2 W $ and the weight of $w^i|_Q$ does not drop below $-(n-k)n^2W$.

   Recall that in $\uu'_{-}$ there are no repeating fair branches; therefore, there are at most $n$-many fair branches. Similar to the base case, 
   there are at most $n$-many simple cycles $\mathbf{c'}$ in $\uu'$ for which $\mathbf{c} = \mathbf{c'}|_Q$ have negative weight. These cycles can collect the minimum weight of $-n^2W$. Due to our assumption $k \neq n$, there is at least one infix $\mathbf{w^1}$ of $\uu'$. Since the weight of $\mathbf{w^1}|_Q$ is at least $ k n^2 W$; $w(\uu) \geq (k - 1) n^2W$; and since the weight of $\mathbf{w^1}|_Q$ does not drop below $-(n-k)n^2W$, the weight of $\uu$ does not drop below $-(n-k+1)n^2W$.

   This finishes the induction, and the proof of the claim.  \qed
   \end{proof}

   This shows that $\sigma$ is a winning player~1 strategy, thereby concluding the proof of the first direction. Now let us prove the other direction. In particular, we show that, if a node is winning for player~2 in $G'$, then it is also winning for player~2 in G i.e., $q \in \Win_2(G') \Rightarrow q \in \Win_2(G)$. 

   \smallskip
   Analogous to the previous part, we consider a positional winning strategy $\pi'$ of player~2 in $G'$ and construct, this time, a \emph{positional} winning strategy $\pi$ for player~2 in $G$. Note that, none of the fair nodes in $G$ belong to player~2 and hence the construction of $\pi$ is quite straight forward: for all $q \in Q_2$, we set $\pi(q) = \pi'(q)$. We define $G'_{\pi'}$ as the subgame restricted to the strategy $\pi'$, as before. 

   For a play $\tau = q^0 q^1 \ldots \in \plays_\pi(G)$, we define its \emph{extension} $\tau'$ in $G'$ inductively as follows: we start from with $q^0$ and $i = 0$, and with increasing $i$,
   \begin{itemize}
       \item If $q^i \not \in Q_f$, append $\tau'$ with $q^{i+1}$; Otherwise,
       \item If the value branch does not exist in $G'_{\pi'}$:
       \begin{itemize} 
           \item If $q^{i+1} = \pi'(q^i_\text{nW+1})$, append $\tau'$ with the fair branch i.e. $q^i_l \cdot q^i_\text{nW+1} \cdot q^{i+1}$, 
           \item Otherwise, append it with the escape branch, i.e. $q^i_r \cdot q^i_{\ESC} \cdot q^{i+1}$.
       \end{itemize}
       \item If the value branch exists in $G'_{\pi'}$:
       \begin{itemize} 
        \item If $q^{i+1} = \pi'(q^i_\texttt{0})$, append $\tau'$ with the zero-value branch i.e. $q^i_l \cdot q^i_\texttt{val} \cdot q^i_{\texttt{0}} \cdot q^{i+1}$, 
        \item Otherwise, append it with the positive-value branch, i.e. $q^i_l \cdot q^i_\texttt{val} \cdot q^i_{\texttt{pos}} \cdot q^{i+1}$.
        \end{itemize}
   \end{itemize}

   We define the tail of $\tau$ as $\tail$ and its extension $\tail'$, analogous to the previous part. We denote the subgraph of $G'_{\pi'}$ restricted to the nodes and edges in $\tail'$ by $G'_{\tail'}$.

   Let $\tau \in \plays_\pi(G)$ be a play in $G$ that conforms with $\pi$. If $\tau$ is not fair, it is automatically won by player~2. Therefore, let us assume that $\tau$ is fair. We will show that the weight of $\tau$ drops unboundedly, i.e. for every $c \in \mathbb{N}$ there exists a prefix $\tau_i$ of $\tau$ whose weight is below $c$. Therefore, $\pi$ is a winning player~2 strategy. 
   
   We show that all simple cycles in $\tail$ have non-positive weight, and moreover, negative weight cycles are visited infinitely often. It will then follow directly that the weight of $\tau$ drops unboundedly. To prove this, we will use the fact that its extension $\tau'$ is a play in $G'_{\pi'}$ and hence all the cycles in $G'_{\tail'}$ have negative weight as $\pi'$ is winning for player~2 for the energy objective in $G'$.

   Our main argument is the absence of fair branches in $G'_{\tail'}$, which we introduce in \Cref{claim:1-fair-energy-absence}.

   \smallskip 
   \begin{claim}\label{claim:1-fair-energy-absence}
     If $\tau$ is fair, there exist no fair branches $(\fair{})$ in $G'_{\tail'}$.
   \end{claim}

   The proof of the claim is idential to the proof of~\Cref{claim:1-fair-absence}. The replacement of simulation branches with value branches in the proof of~\Cref{claim:1-fair-absence} is sufficient for it to serve as the proof of~\Cref{claim:1-fair-energy-absence}. Note that the addition of $-\epsilon$ edges into the weights accumulated from the gagdet branches does not alter the proof at all.

   \noindent Now we show the absence of fair branches in $G'_{\tail'}$ implies that in $\tail$ (i) all simple cycles have non-positive weight, and (ii) a negative weight cycle is taken infinitely often.
   From the definition of extension, (I) for all $q \in Q_f$ only the value branch $\valu{q}$ exists in $G'_{\tail'}$, or (II) for some $q \in Q_f$, only the escape branch $\escape{q}$ exists in $\tail'$. By construction of $\tau'$, case (II) implies that there exists a fair node $q \in Q_f$ that appears infinitely often in $\tau$ but does not take all its fair outgoing edges, in particular $q \to \pi'(q_\text{nW+1}) \in E_f(q)$ infinitely often. This implies that $\tau$ is unfair, contradicting our assumption. Now assume (I) holds and let $\mathbf{c'}$ be a simple cycle in $G'_{\tail'}$. There exists at most $n$-many gadget branches in $\mathbf{c'}$ and each one contributes weight $-\epsilon$ or $0$ to $w(\mathbf{c'})$. So, the minimum weight collected by the gadget branches in $\mathbf{c'}$ is $ - n \cdot\epsilon > - 1$, due to our initial choice of $\epsilon = 1/(n+1)$. Since $w(\mathbf{c'}) < 0$, the non-gadget nodes contribute non-positive weights to $\mathbf{c'}$; that is $w(\mathbf{c}) \leq 0$ for $\mathbf{c} = \mathbf{c'}|_Q$. Therefore, we conclude that all simple cycles in $\tail$ have non-positive weight, i.e. (i) holds. Moreover, since we assumed $\tail$ is fair, each $q$ with $\valu{q}$ in $G'_{\tail'}$ will take its zero-value branch infinitely often; therefore, $\zero{q}$ is in $G'_{\tail'}$. Therefore, a zero-value cycle $\mathbf{c'}$, i.e. a cycle that does not see any other gadget branches but $\zero{}$, is guaranteed to exist in $G'_{\tail'}$. Note that, if there are no value branches in $G'_{\tail'}$, any cycle in $G'_{\tail'}$ is trivially a zero-value cycle. Since $w(\mathbf{c'}) < 0$, this time $w(\mathbf{c}) < 0$ as well. Lastly, since every cycle in $G'_{\tail'}$ is taken infinitely often in $\tail'$, in particular any zero-value cycle $\mathbf{c'}$ is taken infinitely often in $\tail'$; and therefore a negative weight cycle $\mathbf{c}$ is taken infinitely often in $\tail$, i.e. (ii) holds.
   \qed 

   With this, we conclude the proof of the second direction, and that of the theorem.
   
\end{proof}

\subsection{Alternative Proof of Second Part of \Cref{thm:1-fairenergypositionality}}
As in~\Cref{lemma:mean-payoff-2-fair-memoryless}, the sufficiency of positional winning strategies for player~2 in 1-fair energy games follow as a corollary of the proof of~\Cref{thm:energygadgetcorrectness}. Here, using determinacy, we provide an alternative, simpler proof. 

The proof follows from~\cite[Lemma 2]{multiDimEnergyParity} The idea is to show that if player~2 is losing from a vertex $q$ with all positional strategies, it loses from $q$. Due to determinacy, this gives us that if player~2 is winning from $q$, it is winning with a positional strategy.
\begin{proof}
  Let $q$ be a player~2 node, and assume WLOG that it has 2 outgoing edges $q \xrightarrow{w_1} q_1$ and 
  $q \xrightarrow{w_2} q_2$. Assume both $q_1$ and $q_2$ are winning for player~1. Then, there exist minimal initial credits (due to \Cref{eq:energyDeterminacyPlayer1}) $c_1$ and $c_2$ with which player~1 wins from $q_1$ and $q_2$, respectively. Then, the initial credit $\max(c_1 - w_1, c_2 - w_2)$ is sufficient for player~1 to win from $q$.\qed
\end{proof}


\label{endofdocument}
\newoutputstream{pagestotal}
\openoutputfile{main.pgt}{pagestotal}
\addtostream{pagestotal}{\getpagerefnumber{endofdocument}}
\closeoutputstream{pagestotal}
\end{document}